\documentclass[pdflatex,sn-mathphys]{sn-jnl}

\newcommand{\jac}[1]{D\mkern-0.75mu{#1}}
\usepackage{comment}
\usepackage{titlesec}

\jyear{2023}%

\theoremstyle{thmstyleone}%
\newtheorem{theorem}{Theorem}
\theoremstyle{thmstyletwo}%
\newtheorem{remark}{Remark}%

\theoremstyle{thmstylethree}%
\newtheorem{definition}{Definition}%

\usepackage{xcolor}

\begin{document}




\title[Astrocytes as meta-plasticity  and contextually-guided network function]{Astrocytes as a mechanism for meta-plasticity and contextually-guided network function}


\author*[1]{\fnm{Lulu} \sur{Gong}}\email{glulu@wustl.edu}
\author[2]{\fnm{Fabio} \sur{Pasqualetti}}\email{fabiopas@engr.ucr.edu}
\author[3]{\fnm{Thomas} \sur{Papouin}}\email{thomas.papouin@wustl.edu}
\author*[1]{\fnm{ShiNung} \sur{Ching}}\email{shinung@wustl.edu}


\affil*[1]{\orgdiv{Department of Electrical and Systems Engineering}, \orgname{Washington University in St. Louis}, \orgaddress{\city{St. Louis}, \postcode{63130}, \state{MO}, \country{USA}}}
\affil[2]{\orgdiv{Department of Mechanical Engineering}, \orgname{University of California at Riverside}, \orgaddress{\city{Riverside}, \postcode{92521}, \state{CA}, \country{USA}}}
\affil[3]{\orgdiv{Department of Neuroscience}, \orgname{Washington University in St. Louis}, \orgaddress{\city{St. Louis}, \postcode{63110}, \state{MO}, \country{USA}}}



\abstract{
Astrocytes are a ubiquitous and enigmatic type of non-neuronal cell and are found in the brain of all vertebrates.
While traditionally viewed as being supportive of neurons, it is increasingly recognized that astrocytes may play a more direct and active role in brain function and neural computation. 
On account of their sensitivity to a host of
 physiological covariates and ability to modulate neuronal activity and connectivity on slower time scales, astrocytes may be particularly well poised to modulate the
dynamics of neural circuits in functionally salient ways. 
In the current paper, we seek to capture these features via actionable abstractions within computational models of neuron-astrocyte interaction. Specifically, 
we engage how nested feedback loops of neuron-astrocyte interaction, acting over separated time-scales may endow astrocytes with the capability to enable learning in context-dependent settings, where fluctuations in task parameters may occur much more slowly than within-task requirements. 
We pose a general model of neuron-synapse-astrocyte interaction and use formal analysis to characterize how astrocytic modulation may constitute a form of meta-plasticity, altering the ways in which synapses and neurons adapt as a function of time. We then embed this model in a bandit-based reinforcement learning task environment, and show how the presence of time-scale separated astrocytic modulation enables learning over multiple fluctuating contexts. Indeed, these networks learn far more reliably versus dynamically homogeneous networks and conventional non-network-based bandit algorithms. Our results indicate how the presence of neuron-astrocyte interaction in the brain may benefit learning over different time-scales and the conveyance of task-relevant contextual information onto circuit dynamics.
}


\maketitle

\section{Introduction}\label{sec1}

The role of non-neuronal cells (glia) in neural computation has been the topic of increasing interest over the past decade. In the mammalian brain, glia comprise a significant proportion of all cells, comparable to that of neurons.
However, their functional role has traditionally been viewed as one of maintaining the basic physiological needs of neurons \cite{Halassa2007,PEREA2009421,farhy2018astrocytes}. This view has now been repeatedly challenged owing to a decades-long stream of evidence that these cells directly modulate neuronal signaling \cite{murphy2022contextual, nagai2021behaviorally}. The simple premise here is that the computational power of the brain should be conferred by all cells collectively populating the brain, and not merely by neuronal activity \cite{shepherd2018handbook,nagai2021behaviorally,smith1992astrocytes,murphy2022contextual}. That is, the effects of glia on brain function and neuronal activity exist, and hence must matter. This notion opens up richer and more expansive hypotheses regarding the mechanisms underlying brain computation, including ways by which neuromodulation of networks may be achieved and mapped to function.

In the current work, we zero our attention on astrocytes, the most abundant type of non-neuronal cell. Collective work in the field of astrocyte biology has repeatedly provided evidence on the instrumental role of astrocytes in controlling neuronal functions such as synaptic wiring, synaptic activity, synaptic memory, and neuronal excitability \cite{robin2018astroglial,papouin2017septal,henneberger2010long,ma2016neuromodulators,Requie2022,Noh2023,Yuniesky2023,cui2018astroglial,poskanzer2016astrocytes}, reflecting the potential of astrocytes to control key computational loci in the brain. However, directly probing the role of astrocytes in brain computation has been virtually impossible at the experimental level owing to limited knowledge surrounding their rules of engagement and signaling mechanisms, combined with their non-binary rules of ‘excitability’, the multiplex nature of their outputs (i.e. a single astrocyte is capable of a multitude of inhibitory, excitatory or modulatory outputs, contrary to neurons), and the incomplete toolkit available to manipulate them. On the contrary, computational neuroscience provides an ideal playground to probe the role of astrocytes in circuit computation by way of mathematical and algorithmic modeling. On the other hand, the absence of a consensus framework on how to conceptualize astrocyte's contribution to brain computation in a reductionist way has made it difficult to meaningfully abstract astrocyte functions in computational models.
Interestingly, a new hypothesis called ``contextual guidance” was recently introduced that potentially alleviates these issues \cite{murphy2022contextual}.
It posits that astrocytes act as a contextual switchboard that actively conveys information about the environment and physiological state of the organism to neuronal networks. As a result, astrocytes may be a potential active player in mediating neural computation and function. Conversely, accounting for astrocytes, and glia more generally, in neural computation theory may close gaps in how neural circuits learn and implement functions in a manner sensitive to context. For example, an extant issue in theoretical neuroscience pertains to how different but functionally overlapping tasks may be embedded in a single neuronal circuit \cite{driscoll2022flexible,cole2013multi, yang2019task}. Such a scenario would seemingly require mechanisms by which different neuronal dynamical regimes may be learned and then recruited, in a context/task-dependent fashion.
The goal of this paper is thus to introduce computational modeling and analysis to probe how astrocytes may enrich the computational capability of neural circuits toward such objectives.

Astrocytes contain distinct physiological features relative to neurons. They have slow time-scales of activation, on the order or seconds or slower. This latter fact makes them easy to dismiss from the perspective of fast computation. However, these slow time-scales may in fact be a computationally-relevant feature when combined with their uniquely broad spatial scale. A single astrocyte can impinge on hundreds of neurons and synapses.
Indeed, neural network function is often viewed through the lens of synaptic connectivity, wherein specific synaptic `weight' configurations are associated with different tasks \cite{Ahmadian2015,Rivkind2017,Mastrogiuseppe2018,Schuessler2020}. But, surprisingly, the ability of astrocytes to modulate synaptic weight via both Hebbian and non-Hebbian mechanisms is generally unsuspected or disregarded. By providing a mechanism to slowly modulate neural dynamics and synaptic interactions, astrocytes may thus enable functional adaptation according to salient environmental signals or internal circumstances.

Such a framework would represent a shift from common conceptualizations of neural computation that rely on homogeneous neural units, and thus explain how information processing mechanisms may be enacted over different spatial and temporal scales. This, in turn, may better reconcile models of algorithmic learning with the physiological realities of the brain. In fact, recent work has argued that astrocytes may implement a transformer-like model of attention in multi-task adaptation and learning in feedforward architectures \cite{kozachkov2023building}. In \cite{Maurizio2022}, it is shown that neuro-glial interactions can lead in turn to distinct patterns of neural activity in working memory tasks through mean-field network model analyses.  In the current paper, we focus our attention on the \textit{dynamics} of neuron-astrocyte interactions in recurrent network and learning scenarios. The correlation between network dynamics, e.g., vector fields, attractors, etc., and different functions is itself a crucial area of study in theoretical neuroscience \cite{khona2022attractor}.
Furthermore, there is recognition that leveraging the multiple time-scales and heterogeneous structures of recurrent neural networks to design models for learning multiple, sequential, and temporal tasks \cite{Yin2020,Kurikawa2021,Kurikawa20211,Kurikawa20212}. As such, adding astrocytes to traditional recurrent neural network architectures could thus further expand the expressiveness of these networks \cite{Wade2011,Gordleeva2021,Tsybina2022}.  Yet, there remains a considerable gap in our understanding of the dynamics of neuron-astrocyte interactions and how such dynamics may map onto learning and function.

Motivated by the above, our goal is twofold. First, we seek to develop and study a simplified dynamical systems model of neuron-astrocyte interaction in order to gain fundamental insight into how the time- and spatial-scale separation between astrocytes and neurons may enrich the repertoire of neural dynamics and activity. Second, we seek to understand how astrocyte-enriched dynamics may enable learning over disparate time-scales and in context-dependent task scenarios, consistent with the contextual guidance hypothesis outlined above. For the latter, we choose to focus on decision-making problems and reinforcement learning (RL) scenarios, given their relevance and ubiquity in algorithmic learning and prior observations that astrocytes can participate in the encoding of reward information \cite{Doron2022,kang2023astrocyte}.

We proceed to formulate a novel bio-inspired model of neuron-astrocyte interactions, and then embed this model in algorithmic optimization frameworks to solve context-dependent bandit tasks. Our major contributions include the dynamical systems analysis of this model, and understanding glial modulation as a pseudo-bifurcation parameter that can switch neural and synaptic dynamics between different dynamical regimes as a form of meta-plasticity. Here, we use the latter term in a slightly broader sense than traditional definitions \cite{abraham2008metaplasticity}, such that astrocytes will form a `second-order' modulation on the processes that govern neuronal activity as well as the synaptic plasticity rules themselves, resulting in a net re-shaping of the time-evolution of synaptic weights.  We furthermore show that the structure and time-scale separation of astrocytes relative to neurons is enabling in terms of learning non-stationary bandit problems, exceeding the learning performance of well-established algorithms in this domain.

\section{Results}\label{sec2}

\subsection{Neuro-glial interactions constitute a hypernetwork with multi-scale dynamics}\label{subsec21}

We proceed to develop a reduced model of neuron-astrocyte interaction that captures key aspects of neurobiology while enabling fundamental analysis regarding dynamical expressiveness and links to function.

\subsubsection{Neuro-glial structure as a hypernetwork}\label{subsec211}
Classically, biological interactions between neurons, astrocytes, and synapses have been conceptualized in terms of the \emph{tripartite synapse} structure \cite{arizono2020structural,salmon2023organizing} (as shown in Figure  \ref{fig:tripartitesynapse}A). 
Within this framework, astrocytes interact with neurons at synapses, modulating synaptic efficacy \cite{de2023specialized} and controlling synaptic plasticity \cite{panatier2006glia}. Such interactions may occur in a higher-order and `closed-loop' fashion, wherein astrocytes respond to neurotransmitters released during pre- and post-synaptic neuronal activity (see \hyperref[secsi-signalflow]{SI.1}  for detailed description) and this has been the mainstream assumption in past work attempting to model astrocytes. While this description may capture an important dimension of neuron-astrocyte interaction, it is increasingly clear that astrocytic modulation of neuronal activity is more general and multifaceted. The \textit{contextual guidance hypothesis} \cite{murphy2022contextual} espouses that astrocytes not only regulate synaptic activity, but may actively convey exogenous inputs onto said processes. Such inputs may convey contextual information important for circuit function, such as vigilance state, sensory salience, metabolic load, or underlying pathology. As such, astrocytes may actively `control' neural dynamics in a state-dependent manner. These effects may occur not only at the synapse via the release of astrocyte-derived neuroactive transmitters but also at cell bodies via the alteration of ionic conditions, notably potassium levels \cite{PEREA2009421,cui2018astroglial} (Figure \ref{fig:tripartitesynapse}A).

\begin{figure}[t]
    \centering
\includegraphics[width=11cm]{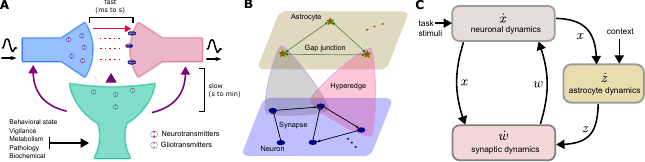}
    \caption{A. In a tripartite synapse, the presynaptic axon and postsynaptic dendrite are surrounded by an astrocyte \cite{baldwin2023astrocyte}, enabling multifaceted effects of neurotransmitters and gliotransimitters. B. A graphical illustration of the neuro-glial hypernetwork: the circles and stars represent neurons and astrocytes respectively; the colored triangles denote the hyperedges and represent the multiplexed intralayer interactions. C. Schematic representation of the feedback interconnections between subsystems in the multi-scale neuro-glial network model.}
    \label{fig:tripartitesynapse}
\end{figure}
The above schema of neuron-astrocyte interactions is difficult to capture as a traditional graphical network representation. As a result, we introduce the framework of a hypernetwork to describe the neuro-glia architecture (see Figure \ref{fig:tripartitesynapse}B for the illustration and detailed description in \hyperref[si_graphical_hypernetwork]{SI.2}).  
We distinguish neurons and astrocytes by representing them on two different layers of the network.  The interlayer relationships are all hyperedges, which embody the ability of astrocytes to modulate neuronal activity at synapses and cell bodies.

\subsubsection{Multi-scale neuronal and astrocytic dynamics}\label{subsec212}
The hypernetwork formulation alone does not capture the full complexity of neuron-astrocyte interaction, as it does not explicitly contain information about the time-scales and dynamics of neuronal and astrocyte activation. 
For this, we introduce a set of ordinary differential equations (ODEs) overlaying the hypernetwork:
\begin{subequations}\label{neuroglialmodel}
    \begin{align}
    &\tau_n\dot{x}_i=-a_ix_i+\sum_{j=1}^nw_{ij}\phi(x_j)+u_i, \quad i=1,...,n, \label{neuroglialmodel1}\\
    &\tau_w\dot{w}_{ij}=-b_{ij}w_{ij}+c_{ij}\phi(x_i)\phi(x_j)+d_{ij}\psi(z_k), \quad i,j=1,...,n, \label{neuroglialmodel2}\\
    &\tau_a\dot{z}_k=-e_kz_k+\sum_{l=1}^{m}f_{kl}\psi(z_l)+g_k,\quad k=1,...,m, \label{neuroglialmodel3} \\
    & g_k = h_k\phi(x_i)\phi(x_j)+v_k. \label{neuroglialmodel4}
    \end{align}
\end{subequations}
These dynamical equations are based on firing rate descriptions of neural activity (see \hyperref[sec4]{Methods} for modeling details). 
Here, $x_i$ describes the rate of the neuron $i=1,...,n$, $w_{ij}$ is the weight of the synapse (i.e., the synaptic efficacy)  between neurons $i$ and $j$, and $z_k$ represents the activity (abstracted from Calcium activity) of astrocytes $k=1,...,m$.  Here we emphasize that $z_k$ embeds a graded but non-linear transformation between the inputs to astrocytes and their output onto neurons.  There exist many models for describing the dynamics of neurons, and the one we use is, in essence, a continuous-time rate-based recurrent neural network (RNN) \cite{funahashi1993approximation}. For the edge weights between neurons, we prescribe a Hebbian plasticity rule wherein weight changes are dependent on the correlation $\phi(x_i)\phi(x_j)$. The signal $u_i$ conveys external inputs onto neural dynamics.

To distinguish astrocytes from neurons, we use a different activation function (i.e., $\psi(\cdot) \neq \phi(\cdot)$) and, most crucially, will assume that the time-scale $\tau_a$ is slower than that of neurons. Specifically,  a larger value of $\tau_n$, $\tau_w$, and $\tau_a$ implies a slower rate of time-evolution \cite{kuehn2015multiple} of the associated activity variables. Thus, the multiple time-scale feature of neural-glial processes is readily captured in equations \eqref{neuroglialmodel}, with a suitable choice of the values of these parameters. 
Completing the model, $f_{kl}$ denotes interactions between astrocyte $l$ and $k$, allowing for potential gap junctions-mediated communication between neighboring astrocytes \cite{pannasch2012astroglial}.  An important feature of the model is that astrocytes may be sensitive to contextual information in accordance with \cite{murphy2022contextual}, via $c_k$.  Here, we postulate two forms of context as specified in \eqref{neuroglialmodel4}. First, we consider a `circuit' context, such that the astrocyte may have a sensitivity of second-order neuronal activity via the coefficient $h_k$. Second, we formulate an external context, motivated by the contextual guidance hypothesis, conveyed by the exogenous `contextual signal' $v_k$. Such a signal may originate, for example, from the sensory periphery. The neuronal exogenous input $u_i$ may also contain such contextual information.

The model above attempts to balance expressiveness, interpretability, and tractability.  In particular, we have not fully captured the spatial scale distinctions of astrocytes relative to neurons here, since we restrict ourselves to only the case of two neurons within the domain of a single astrocyte. We have, however, captured several important features of the astrocytic contextual guidance hypotheses: (i) the presence of multiple, nested loops of feedback between neurons and astrocytes, providing a diversity of mechanisms by which contexts can propagate through astrocytes and affect neuronal activity, and (ii) the potentially orders-of-magnitude separation in time-scales between neuronal activity and astrocytic modulation thereof. 
Unlike previous abstractions such as \cite{Maurizio2022}, we do not assume spike-like dynamics within astrocytes.  In total, our astrocytes: (a) produce slow, graded activity (as a surrogate for $Ca^{2+}$) that (b) modulates neuronal excitability and synaptic plasticity and (c) is responsive to circuit and external context via feedforward and feedback signaling paths.
It is of note that the above neuro-glial model is well-behaved from a dynamical systems perspective since solutions exist, are unique, and are restricted to a bounded subspace (see \hyperref[secsi1]{SI.3}).

From a systems-level perspective, the dynamics of the neuro-glial network can be understood as the interaction between three subsystems, forming two closed-loops as shown in Figure \ref{fig:tripartitesynapse}C. The first closed-loop consists of the subsystem of neurons (\ref{neuroglialmodel1}) and synapses (\ref{neuroglialmodel2}). The second closed-loop involves the subsystem of astrocytes (\ref{neuroglialmodel3}), which transfers information from neurons to synapses.  By forming these closed-loops, the astrocytic process not only directly modulates synaptic plasticity based on neural activity but also indirectly modifies synaptic connections, shaping the dynamics of the network as a whole. This mechanism can facilitate the formation and evolution of attractors (e.g., fixed points) in the neural subsystem state space, as elaborated below.

\subsubsection{Glial modulation acts as a pseudo-bifurcation parameter that enables meta-plasticity and rapid changes in circuit dynamics}\label{subsec22}
To analyze the dynamics of \eqref{neuroglialmodel}, we reduce it to its simplest motif, i.e., the interaction of two neurons and a single astrocyte.  
Here, we assume that the neurons form a reciprocal excitatory-inhibitory loop, itself a common canonical motif for cortical interactions between pyramidal and inter-neurons.
\begin{figure}[htbp!]
    \centering
\includegraphics[width=10cm]{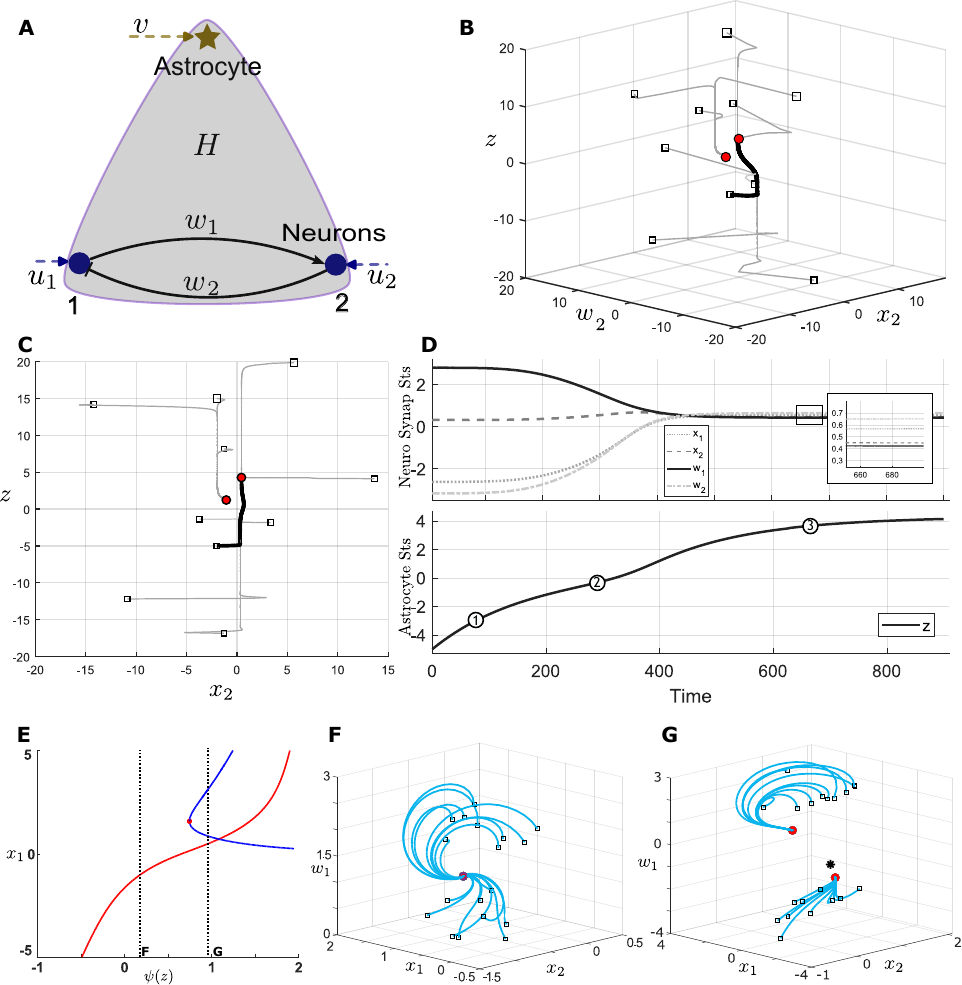}
 \caption{Neuro-glial network motifs and dynamic properties. A. Graphical representation of the network motif;
 $u_1$, $u_2$, $v$ include inputs from other nodes of the hypernetwork as well as those from external sources. B, C. Several examples of phase curves in the state space $(x_2,w_2,z)$ of the network motif system. The parameter conditions are $a_1=0.7$,
$a_2=0.6$,
$b_1=1.6$,
$b_2=1.7$,
$c_1=12$,
$c_2=-10$,
$d_1=-4$,
$d_2=5$,
$e=0.6$, $h=6$, and $\tau_1=\tau_2=0.01$, $\tau_3=1$.
 The system has three fixed point points, of which two are stable (red dots). The system dynamics converge to these two stable fixed points.  D. Trajectory associated with the thick phase curve from B, C. illustrating two stationary regimes (indicated by 1 and 3 in the figure). E. depicts the bifurcation diagram of the neural dynamics with respect to the astrocyte output $\psi(z)$, where the red curve shows that one branch of fixed point always exists, while the blue curve shows how the other branch of fixed points changes via the saddle-node bifurcation. F, G. Vector fields of the neuronal-synaptic dynamics to either side of the saddle-node bifurcation.}
\label{fig:networkmotif}
\end{figure}
From \eqref{neuroglialmodel}, the neuron-astrocyte motif amounts to a set of $5$ ODEs:
\begin{equation}\label{networkmotifdyanmics}
 \begin{aligned}
 &\tau_1 \dot{x}_1=-a_1x_1+w_2\phi(x_2)+u_1(t)\\
 &\tau_1 \dot{x}_2=-a_2x_2+w_1\phi(x_1)+u_2(t)\\
 &\tau_2 \dot{w}_{1}=-b_{1}w_{1}+c_{1}\phi(x_1)\phi(x_2)+d_{1}\psi(z)\\
 &\tau_2\dot{w}_{2}=-b_{2}w_{2}+c_{2}\phi(x_1)\phi(x_2)+d_{2}\psi(z)\\
 &\tau_3 \dot{z}=-ez+h\phi(x_1)\phi(x_2)+v(t).
 \end{aligned} 
\end{equation}
The dynamics of this system are asymptotically bounded  (see \hyperref[secsi2]{SI.4}). 
Within this bounded set, the motif may exhibit a unique fixed point, or multiple fixed points, depending on parameterization. 
Figure \ref{fig:networkmotif}B,C shows the case of three fixed points under the assumption that astrocytes evolve at a time-scale two orders of magnitude slower than neurons and synapses (i.e., $\tau_3 = 100 \tau_1, \tau_2$). 
Figure \ref{fig:networkmotif}D illustrates the time evolution of a specific trajectory within this landscape. 
As expected, $z$ evolves much slower than the other variables. Notably, this slowly-changing astrocytic activity variable seems to drive neural variables to transit between nearly stationary regimes, suggesting that astrocytes can systematically `control' stationary neural activity.

In order to understand this phenomenon in more detail, we performed a singular perturbation analysis (see \hyperref[secsi3]{SI.5})  to better clarify the mechanisms by which astrocyte signals may be modulating neural dynamics. 
This analysis treats the astrocyte state as a fixed parameter, premised on its relatively slow evolution relative to the neural dynamics.  
We can then study how this parameter affects the vector field and attractor landscape of the neural subsystem.  
Figure \ref{fig:networkmotif}E provides the pseudo-bifurcation diagram of the above motif by showing the position of the fixed points in the $x_1$-dimension as a function of the $\psi(z)$. When $\psi(z)$ is small, there is only one fixed point (the red line). When $\psi(z)$ is large, the neural subsystem manifests three fixed points by means of a saddle-node bifurcation. 
In other words, at the bifurcation point, there is a fundamental change in the shape of the neuronal-synaptic vector field and hence dynamics. Thus, astrocytic modulation can drastically alter the flow of neuronal and synaptic activity as a function of time.  We hypothesize this mechanism may be particularly powerful for the contextual guidance premise as it may enable astrocytes to reshape the dynamics of synaptic adaptation and hence neural computation, based on exogenous contextual signals, e.g., via $v(t)$. Thus, astrocytes form, in essence, a pathway for context-guided meta-plasticity and targeted neuromodulation. Below, we probe this hypothesis within the reinforcement learning task paradigm.

\subsection{Neuro-glial networks are able to learn context-dependent decision-making problems}\label{subsec23}
We apply the proposed multi-scale neuro-glial network model to context-dependent decision-making problems.
We focus specifically on multi-armed bandits (MABs), a well-known class of reinforcement learning problems, wherein an agent aims to maximize its cumulative reward over time by selecting actions (arms) from a set of available options \cite{sutton2018reinforcement}. MABs find applications in various domains, including recommendation systems, clinical trials, and cognitive tasks in neuroscience, as they provide a powerful framework for decision-making under uncertainty \cite{slivkins2019introduction}.  While well-studied, this class of problems nonetheless poses persistent challenges when environments are non-stationary. Our prevailing hypothesis is that the disparate time-scales of signaling emanating from astrocytes can enable learning in such settings.

A standard MAB assumes a constant environment, in which the probabilities of reward associated with different arms are stationary. Our goal, however, is to study the capacity of our proposed neuro-glial networks, by virtue of their time-scale separation, to learn in non-stationary and/or context-dependent settings. Thus, we designed both stationary and non-stationary Bernoulli bandit environments (see Figure \ref{fig:main_result1} and Multi-armed bandit tasks in \hyperref[sec4]{Methods}) within which to evaluate learning efficacy.

\subsubsection{Learning metrics}\label{subsubsec231}
In MABs, a common figure of merit is the (pseudo) cumulative regret, which is defined specifically in Bernoulli bandits by 
\begin{equation}\label{regretdefinition1}
    R_T= \sum_{t=1}^{T} (\max_{a_i\in \mathcal{A}} \mu_i -\mathbb{E}[r_{t}]),
\end{equation}
where $T$ is the total rounds, $\mu_i$ is the mean of the action $a_i$, which belongs to the action set $\mathcal{A}$, and $r_{t}$ is the reward derived by the agent at trial $t$ with $\mathbb{E}[\cdot]$ denoting the expected value.
A lower value of \eqref{regretdefinition1} indicates less accumulated loss and equivalently higher accumulated reward. Additionally, we consider the convergence speed of the algorithm, which measures the time taken by the agent for $R_T$ to reach an optimal value. Faster convergence is generally desirable as it signifies more efficient learning by the algorithm.

\subsubsection{Learning algorithm architecture}\label{subsubsec232}
In order to evaluate the proposed neuro-glial model in these tasks, we require a learning/optimization method. For this purpose, we make several implementation assumptions.  First, we assume that the network emits an output via a softmax operation, a typical form of network readout in neural network architectures. Second, we assume that networks have access to a signal that contains information about the environmental context (e.g., a change in arm probabilities, without overtly specifying the probabilities themselves). 
Upon this architecture, we deploy a reinforcement learning method to optimize all parameters of the model (see \hyperref[sec4]{Methods}).
 The architecture of our learning algorithm is depicted in Figure \ref{fig:main_result1}. 
Briefly, during a typical learning episode, the network outputs a policy for action selection, i.e., a probability distribution over the possible actions (at the output of the softmax). The bandit environment provides a reward to the agent in response, which is then fed into an analytical loss function, for which a gradient can be defined and hence network parameters updated.
Crucially, this learning paradigm is agnostic to the specific network being learned, i.e., we can train vanilla RNNs and other architectures with the exact same methodology. This will allow us to make direct comparisons between the proposed neuro-glial network and other standard neural networks.

\begin{figure}[t]
    \centering
\includegraphics[width=11cm]{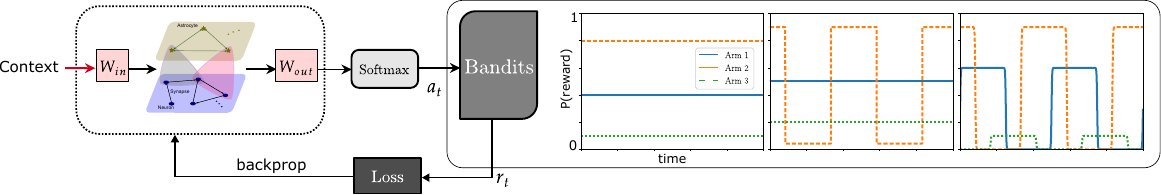}
    \caption{Architecture of the learning algorithm. The three plots on the right represent a stationary Bernoulli bandit scenario where the arm means remain  $(0.4, 0.8, 0.1)$ constantly over time,  a flip-flop non-stationary Bernoulli bandit where Arm 2's mean alternates between $0.92$ and $0.042$, and a smooth-change non-stationary Bernoulli bandit where all arm means change according to a smooth periodic function, respectively.
The left figure shows the architecture of the learning algorithm. }
    \label{fig:main_result1}
\end{figure}

\subsubsection{Performance comparison}\label{subsubsec233}
We conducted a comprehensive learning performance analysis of the proposed neuro-glial network in comparison to other neural network architectures (vanilla RNN, LSTM, GRU), all trained the same way using the above method. In addition, we also deployed traditional algorithms for solving bandit problems, the  Upper Confidence Bound (UCB) and Thompson Sampling (TS) methods. 
The specific learning procedures for all neural network-based methods are similar, as described in Section \ref{subsubsec232}.

\paragraph{Stationary case.} Figure \ref{fig:main_result1}E,F illustrates the comparison of the learning performance of 
 different methods (Neuro-glial, LSTM, TS, vRNN, GRU, UCB) in a stationary bandit task with arm probability settings of $(0.4, 0.8, 0.1)$. Each method requires exploration of the environment, resulting in high regret during the initial time steps. However, all methods eventually converge with comparable rates and cumulative regret of the same order of magnitude. In particular, the neuro-glial architecture performs similarly to the other network-based implementations in this case. Single-run simulation results show that the neuro-glial method uses less time to converge (see  \hyperref[si_regrets_time]{SI.6.1}). In addition, this method tends to be robust as the tasks become more challenging due to the small distance between arm probabilities (see  \hyperref[secsirobust]{SI.6.2}).

\paragraph{Non-stationary case.}
However, in the presence of non-stationary, the neuro-glial architecture displays significant gains in capability. 
Indeed, these networks can achieve almost stationary regrets over time as shown in Figure \ref{fig:main_result_2}B.  In contrast, other methods consistently result in escalating regrets. It is important to emphasize again that the setup for learning here is identical across all networks. 
These results are consistent across different non-stationary scenarios (see  \hyperref[si-detail_non-stationary]{SI.6.3}). In addition, similar learning performance is observed in scenarios with a different number of actions (see  \hyperref[secsi-8actions]{SI.6.5}). 
These observations suggest that the neuro-glial network is able to leverage contextual signals and adapt its actions to the changing environment.

\begin{figure}[t]
    \centering
\includegraphics[width=10cm]{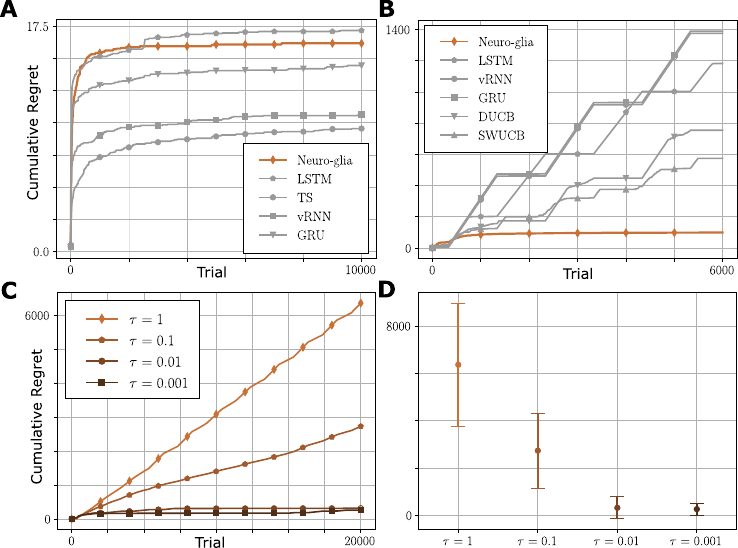}
    \caption{Learning performance. Performance comparison of the neuro-glial method relative to other learning methods for A. stationary and B. flip-flop bandit environments. C,D. neuro-glial learning performance for different time-scale separation.}
\label{fig:main_result_2}
\end{figure}

\subsection{Time-scale separation is necessary for context-dependent learning}\label{sec235}

In order to probe the mechanisms by which the neuro-glial network achieves context-dependent learning, we first focus on the time-scale separation between neurons and astrocytes. In our analysis above, we showed how astrocytic modulation may function, in essence, as a form of meta-plasticity wherein the time-scale separation enabled pseudo-bifurcations that could allow neuronal dynamics to traverse different functional regimes.
The question at hand is whether this mechanism confers utility for context-dependent learning. 
To assess this, we varied the time-scale separation (via $\tau$) between astrocytes and neurons in our network, to probe the impact of this feature on learning performance.
As shown in Figure  \ref{fig:main_result_2}C (also \hyperref[secsi_timescale]{SI.6.4}), different $\tau$ have significant impacts on learning performance, to the extent that without time-scale separation learning simply does not occur. This is seen for the case $\tau=1$, in which astrocytes and neurons have the same time-scale. Here, the cumulative regret does not converge. When $\tau=0.1$, the agent can sometimes achieve stationary asymptotic cumulative regret. This learning performance improves for greater time-scale separation. For $\tau \leq 0.01$ (that is, time scale separation greater than $2$ orders of magnitude), the agent can always adapt to the environment. Moreover, for greater time-scale separation with smaller values of $\tau$, there is less variability in the asymptotic regret (see Figure \ref{fig:main_result_2}D).

To understand the mechanism underlying this effect, we more closely examined the learning dynamics of individual model instances over the different $\tau$ values,  especially the $\tau=1$ and $\tau = 0.01$ cases.  As shown in Figure \ref{fig:trained_network_analysis}A, in the case of $\tau=1$, the network is able to learn solutions in each context; however, upon switching, regret again accumulates, indicating an overwriting of prior strategies as comparable to the phenomenon of catastrophic forgetting. On the other hand, neuro-glial networks with time-scale separation are able to reliably learn the flip-flop bandit, indicating that they are able to gradually associate the contextual information with the environment and protect previously learned trajectories. As shown in Figure \ref{fig:trained_network_analysis}B, the astrocyte-mediated meta-plasticity appears to be engaged during the process of learning. Specifically, we projected the trial-wise network activity along population vectors associated with astrocytes ($PC_z$) and synaptic weights ($PC_w$). We observed that during learning, the network forms distinct synaptic trajectories that asymptotically approach a fixed weight configuration. The time-scale separation between astrocytic and synaptic activation is apparent when tracing the initial stages of the trajectories. Furthermore, the astrocyte output is less sensitive overall to learning, likely an important factor in preventing the context-wise overwriting of prior dynamics (see also \hyperref[sec3]{Discussion}). 
\begin{figure}[t]
   \centering
\includegraphics[width=10cm]{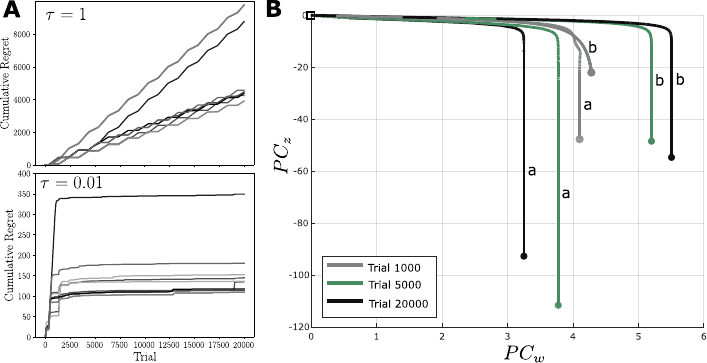}
   \caption{A. Single learning traces for $\tau=1$ and $\tau=0.01$, highlighting the role of time-scale separation in enabling RL over contexts. B. Astrocyte and synaptic activity projections for both contexts (indicated as $a$ and $b$) in early, middle, and late phases of learning, highlighting the formation of distinct synaptic weight trajectories.}
   \label{fig:trained_network_analysis}
\end{figure}

\section{Discussion}\label{sec3}

\subsection{Toward a fuller accounting of brain circuit dynamics}
In this paper, we have examined the potential role of neuro-glial interactions in context-dependent learning, with a specific focus on reinforcement-based bandit problems. We began by forming a simplified model of such interactions in the form of a dynamical system, leveraging canonical descriptions of neural firing rate activity and several abstractions of astrocytic activity and modulation that are based on extant neurobiological theory. In particular, we simplified the dynamical description of astrocytes and focused on two key aspects: (i) their orders-of-magnitude time-scale separation from neurons, and (ii) their modulation of synaptic processes. 
Our goal was to understand whether these aspects of neuro-glial interaction, which are known to exist in the brain, matter for network computation.

\subsection{Contextually-guided meta-plasticity}
From this perspective, our analysis indicates the potential for astrocytes to reshape neural and synaptic vector fields in quite significant ways, such as in the formation of multiple stationary regimes of activation and changing the geometry of synaptic weight evolution. Perhaps most notably, astrocytes can modify the dynamics of synaptic plasticity, effectively switching the network between slow and fast weight adaption regimes. This forms a powerful mechanism by which astrocytes can use external and internal contextual information \cite{murphy2022contextual} to shift networks between different modes of learning, which we view as a form of meta-plasticity in the wide sense of that term. 

One important assumption we have made in this work is the use of a contextual signal that is accessible by astrocytes and neurons, with the premise that such a signal may embed task-relevant information and/or other circuit contexts, which is highly consistent with the body or work showing astrocytes ability to detect and respond to functionally salient physiological covariates such as neuromodulators (e.g., dopamine), hormones (e.g., glucocorticoids), or local cytokines. Our abstraction of this signal may be viewed as overly strong, insofar as it presents `clean' context information to the network. From this perspective, we emphasize that all our alternative architectures, and especially the neuro-glial model without time-scale separation, had access to this information. Thus, it is not merely the presence of contextual signaling that enables learning, but the specific dynamical mechanisms by which this information alters neuronal and synaptic dynamics that augment the learning performance.

\subsection{Astrocytic activity as a stabilizer of catastrophic forgetting}
Catastrophic forgetting is a phenomenon in artificial neural networks that arises when networks are tasked with learning multiple tasks sequentially \cite{kirkpatrick2017overcoming}. In this scenario, it often is the case that previously encountered tasks are `overwritten'  when the algorithmic optimization (i.e., learning) strategies are deployed to update the network parameters/weights to meet new task demands. Our results indicate that astrocytic modulation of neuronal and synaptic dynamics may mitigate catastrophic forgetting.  Here, we believe that the slow time-scale of astrocytes is instrumental in protecting previously learned network outputs upon encountering of a new context. As described above, the slow activation of astrocytes makes them generally less sensitive to parametric adjustment relative to neurons and synapses. Thus, their effects are more stable context-to-context. Furthermore, as we have seen, astrocytes have the effect of controlling neuronal and synaptic dynamics, such that those faster processes can occupy distinct regions of state space depending on astrocytic modulation. The combination of these two phenomena means that astrocytes can effectively insulate the learned trajectories/dynamics of one context from overwriting when learning is engaged for a subsequent context. These findings underscore the importance of dynamical heterogeneity in the brain and support the functional advantages that astrocytes may confer.

Clearly, an important next step for these models will be to validate them with appropriate experiments. As mentioned in the Introductions, tools for \textit{in vivo} study of astrocyte function have lagged relative to those for neurons. For those tools that do exist, e.g., to disrupt astrocyte function \cite{shen2021chemogenetic}, studies such as ours can identify salient behavioral paradigms within which experiments may be conducted.
Most directly, the model suggests that astrocytes contribute to learning in context-dependent or, potentially, multi-task settings. 
One can easily imagine that these tools will soon be deployed to test more formally the role of astrocytes such scenarios. For example, by examining the learning efficacy of rodents engaging multi-arm bandit paradigms \cite{ohta2021asymmetric}.

\subsection{Insights into algorithmic learning systems}
While our goal in this paper has been to explore new theories regarding the potential significance of neuro-glial interactions in the brain, it is nonetheless interesting to consider the implications of these results in the domain of algorithmic systems. We have already commented on the fact that traditional algorithmic methods of learning bandit tasks have difficulty in context-dependent settings, even in the presence of informative signaling. This begs the question of whether neuro-glial type architectures may have utility beyond the bandit/reinforcement learning settings. 

In this regard, there certainly exist recurrent neural networks designed to deal with multiple time-scale features, notably LSTMs \cite{hochreiter1996lstm} and hierarchical RNNs \cite{chung2016hierarchical}. The LSTM has an internal memory cell state that enables it to deal with tasks that involve long-term dependencies. In hierarchical RNNs, multiple layers of RNNs are stacked on top of each other, where each layer captures information at a different level of temporal abstraction. The lower layers focus on short-term dependencies, while the higher layers focus on longer-term dependencies. The multi-scale neuro-glial network considered here is in the form of a feedback-connected multi-layered network with different embedded time-scales, and hence may blend the different features of these extant machine learning architectures. It is thus possible that this framework may be extendable to other machine learning domains, especially ones involving disparate time scale requirements such as meta-learning \cite{hochreiter2001learning,wang2016learning,wang2018prefrontal}.

\section{Methods}\label{sec4}

\textbf{Multi-scale modeling of neuro-glial network dynamics}\\
In general, neural dynamics can be described by recurrent neural network models. Here, we consider the biology-inspired continuous-time RNN (CTRNN) \cite{funahashi1993approximation, beer1995dynamics}. Consider a group of $n$ neurons where each neuron is connected to some other neurons via synapses. Let $x_i \in \mathbb{R}$ be the state of the unit $i$, which denotes the mean membrane potential of the neuron. Then, the model of CTRNN is defined by ODEs
\begin{equation}\label{CTRNN}
    \tau_n\dot{x}_i=-a_ix_i+\sum_{j=1}^nw_{ij}\phi(x_j)+u_i, \quad i=1,...,n,
\end{equation}
where $\tau_n>0$ and $a_i>0$ are the time constant and decaying parameter respectively, and $u_i$ is the external input to unit $i$. $\phi(x_j)$ is the activation function. It is noted that each unit $i$ collects the outputs $\phi(x_j)$ (i.e., short-term average firing frequency) from all the connected neural units in the network, weighted with the synaptic connection coefficients $w_{ij}\in \mathbb{R}$, where the positive or negative $w_{ij}$ indicates an excitatory or inhibitory synapse respectively.

Synapses are capable of modifying their strength via synaptic plasticity, which is
usually formulated as a learning rule where the change of a synaptic strength $w_{ij}$ depends on the correlation between the firing rate of a  presynaptic neuron $j$ and the firing rate of the postsynaptic neuron $i$. We consider the Hebbian learning rule: the weight between two neurons strengthens when they are correlated, and weakens otherwise. This rule is defined mathematically by the equation \cite{Gerstner2002}
\begin{equation}\label{hebb}
\tau_w\dot{w}_{ij}=-w_{ij}+c_{ij}\phi(x_i)\phi(x_j),
\end{equation}
where $b_{ij}>0$ is the decaying parameter; $\tau_w>0$ is the time constant; $c_{ij}\in \mathbb{R}$ is a parameter which indicates an existing synaptic connection when it is non-zero.  When $c_{ij}$ takes a positive value, \eqref{hebb} is called the \emph{Hebbian learning}, and the case with $c_{ij}<0$ is \emph{anti-Hebbian learning}. 

Similar to neurons, astrocytes can establish network connections within the central nervous system through gap junctions \cite{pannasch2011astroglial,pannasch2012astroglial}. Biophysically, the increase in calcium ion $\text{Ca}^{2+}$  levels within individual glial cells can propagate to neighboring glial cells over long distances, forming $\text{Ca}^{2+}$  waves \cite{Haydon2001}. 
Current mathematical models for astrocytes are excessively complex and not easily translatable for analytical and computational purposes. Therefore, we propose a simplified glial network model to describe astrocyte dynamics. This model is constructed based on the analogy of neural networks, following the framework outlined in \cite{DePitta2020}.

Consider a group of $m$ astrocytes. Let $z_k \in \mathbb{R}$ be the state of astrocyte $k$ which denotes the activity of calcium wave. For the glial node $z_k$, we assume the output of astrocyte calcium wave is similarly defined by an activation function. To distinguish it from the neuron, we use a different function, for instance, the hyperbolic tangent function $\psi(z_k)=\tanh(z_k)$. Then, in the absence of neuro-glial interactions, the dynamics of $z_k$ is described by 
\begin{equation}\label{glialsystem}
    \tau_a\dot{z}_k=-e_kz_k+\sum_{l=1}^{m}f_{kl}\psi(z_l)+ v_k,\quad k=1,...,m, 
\end{equation}
where $\tau_a$ is a constant time parameter; $f_{kl}$ denotes the  weight of the connection from astrocyte $l$ to $k$; $v_k$ captures other external inputs. 
The usage of this phenomenological model can be justified with analogous arguments in \cite{Maurizio2022}, where a neuronal leaky integrate-and-fire model is used for astrocytes. Such a model is easy to be modified to incorporate the neuro-synapse-glial interactions and greatly facilitates the numerical and analytical investigation as shown in Sections \ref{subsec22}, \ref{subsec23}.

Stacking all the equations of neurons, synapses, and astrocytes together, we will arrive at the mathematical model for the neural-glial network as a whole. 
\begin{subequations}\label{neuroglialdynamics}
    \begin{align}
    &\tau_n\dot{x}_i=-a_ix_i+\sum_{j=1}^nw_{ij}\phi(x_j)+u_i, \quad i=1,...,n,\\
    &\tau_w\dot{w}_{ij}=-b_{ij}w_{ij}+c_{ij}\phi(x_i)\phi(x_j)+d_{ij}\psi(z_k), \quad i,j=1,...,n,\\
    &\tau_a\dot{z}_k=-e_kz_k+\sum_{l=1}^{m}f_{kl}\psi(z_l)+h_k\phi(x_i)\phi(x_j)+v_k,\quad k=1,...,m,
    \end{align}
\end{subequations}
where the additional terms $d_{ij}\psi(z_k)$ and $h_k\phi(x_i)\phi(x_j)$ with $d_{ij}, h_{k}\in \mathbb{R}$ are present to capture the high-order interaction between neurons, astrocytes and synapses according to the description in tripartite synapse structure.
In system \eqref{neuroglialdynamics}, there are $n$ and $m$ equations for $x$ and $z$ respectively. The number of synaptic connections is flexible and denoted by $o$ with $ m\leq o \leq n(n-1)$. Therefore, the dimension of system \eqref{neuroglialdynamics} is actually $(m+n+o)$. 

It is known that the activities of neurons, synapses, and astrocytes evolve on different time-scales. Neural firing occurs in milliseconds, synapse plasticity changes at a slower speed, and astrocyte processes take even longer, ranging from seconds to minutes. These varying time-scales significantly impact information processing in neural-glial interactions. To investigate the effects of these differences, we need to set the time-scale parameters, denoted as  $\tau_n$, $\tau_w$, and $\tau_a$, to different values. 
To make the speeds of the evolution of these variables distinguishable,
we have the assumption: $0<\tau_n \ll \tau_w \ll \tau_a$, with $\ll$ indicating the former entity is much smaller than the latter.  As the main goal of this work is to study neuron and astrocyte computation, we set $\tau_n = \tau_w$ for simplicity when applying the neuro-glial model to solving the tasks.

\noindent
\textbf{Dynamic context-dependent multi-armed bandit tasks}\\
  In the setting of a stochastic MAB, there is a set of actions (arms) $\mathcal{A}$ to choose from, and the bandit lasts $T$ rounds in total. In each round $t$, an agent (decision-maker)  chooses one action $a_t\in \mathcal{A}$ and obtains a reward $r_t$. The goal of the agent is to optimize the accumulated reward, i.e., $ \max_{a_t\in \mathcal{A}} \sum_{t=1}^{T}r_t$.  We consider the Bernoulli bandits which belong to stochastic MABs.
In the context of Bernoulli bandits,  the reward of each action is binary, either $1$ or $0$ depending the outcome is a success or failure.  The reward $r_i$ of the $i$-th action is drawn from a Bernoulli distribution, i.e.,
\[r_i \sim \text{Bernoulli}(\mu_i),~~i=1,...,n,\]
where $\mu_i\in [0,1]$ is a constant denoting the mean of the distribution. Different actions have different $\mu_i$ where a larger value represents a higher probability of the successful outcome and thus a higher expectation of the reward.  The reward sequence up to time $T$ is a random process
\begin{equation}
   \{ r_t \sim \{\text{Bernoulli}(\mu_i)\}_{i=1}^{n},~~ t=1,...,T.\}
\end{equation}
In the Bernoulli bandit, the goal of optimizing the accumulated reward is equivalent to minimizing the cumulative regret \eqref{regretdefinition1}. 
The standard Bernoulli bandit is stationary where all $\mu_i$ are fixed over time. 
In addition to the stationary case, we further consider non-stationary variants by making the means changeable and time-dependent. Two subcases are considered in this work:
\begin{enumerate}[1.]
\item Flip-flop switching:  the means $\mu_i$ of actions remain constant for a certain period of time, and then abruptly transit to different values $\mu_i'\in [0,1]$ at certain time instants.
\item Smooth changing: the means change according to a continuous function of time. Here, we use the periodic function 
\begin{equation}
    \mu_i(t)=\mu_i^*   S\left(Q\sin \left (\frac{2\pi t}{P} + \frac{2\pi i}{n} \right)\right), 
\end{equation}
where $\mu_i^*$ is a fixed value in $[0,1]$; $S(\cdot)$ is the sigmoid function; $P$ is used to control the period of this function and the term $\frac{2\pi i}{n}$ makes that the action with the highest expected reward can change between the available actions over time. When $Q$ is large, this type of function is dominated by an approximately constant value, such that it looks like a smooth square wave. We set $P$ and $Q$ to $10000$ and $100$ respectively.
\end{enumerate}

In dynamic bandits, when the arm means change over time and the action with the highest mean switches,  contextual information can be revealed to the agent. This contextual information represents the changes in underlying contexts. 
Therefore, the tasks we considered become context-dependent. 
We define the contextual signals as a scalar in all the simulations presented in this work. However, it is important to note that these signals can also be expanded into a multi-dimensional vector to accommodate more general settings.

\noindent
\textbf{Discrete-time neuro-glial network}\\
For simplification, we assume that the self-decay parameters are all one and the time-scales of neurons and astrocytes are the same. Then, the neuro-glial network model without inputs can be rewritten in the compact form
\begin{equation}\label{compactformneuroglia}
\begin{aligned}
 &\tau \dot{x}=-x+W\phi(x)\\
 &\tau \dot{W}=-W+C\Phi(x)+D\psi(z)\\
 &\dot{z}=-z+F\psi(z)+H\Phi(x), 
\end{aligned}
\end{equation}
where $x=[x_1,...,x_n]^\top$ and $z=[z_1,...,z_m]^\top$ are state vectors for neurons and astrocytes; $W=[w_{ij}]$ is the matrix for synapse weights and $\dot{W}$ denotes the element-wise derivative of $W$; $\phi(x) =[\phi(x_1),...,\phi(x_n)]^\top$  and $\psi(z)=[\psi(z_1),...,\psi(z_m)]^\top$ are  vectors of activation functions while $\Phi(x)$ is the flatten vector of the matrix $[\phi(x_i)\phi(x_j)]$; $C$, $D$, $F$, and $H$ are the parameter matrices with corresponding entries. 

In \eqref{compactformneuroglia}, we have set the time constant for astrocytes to the unit, while time constants for neurons and synapses are both $\tau \ll 1$. In this way, $\tau$ is dimensionless and represents the time-scale difference rate between neurons and astrocytes. 
Note that \eqref{compactformneuroglia} can be rewritten equivalently by a change of time, so that $\tau$ appears on the right hand side of $\dot{z}$.

By using the first-order Euler discretization method, we can transfer the continuous-time neuro-glial model to the discrete-time approximated version
\begin{equation}\label{discretengmodel}
\begin{aligned}
 &x_t=(1-\gamma)x_{t-1}+\gamma W_{t-1}\phi(x_{t-1})\\
 &W_t=(1-\gamma)W_{t-1}+\gamma (C\Phi(x_{t-1})+D\psi(z_{t-1}))\\
 &z_t=(1-\gamma \tau)z_{t-1}+\gamma \tau(F\psi(z_{t-1})+H\Phi(x_{t-1})),
\end{aligned}
\end{equation}
where $\gamma$ is the discretization step size. In the following simulations, $\gamma$ and $\tau$ are set to be $0.1$ and $0.01$ respectively. We  use the sigmoid function $\phi(x)=1/(1+e^{-x})$ and the hyperbolic tangent function $\psi(z)=\tanh(z)$ for neural and glial layer in the simulations.

We incorporate this discrete time neuro-glial model as the hidden layer within the entire learning network, where a pair of linear input and output layers are placed before and after the hidden layer according the convention.
The input  $I\in \mathbb{R}^{\lvert u\rvert}$ and the output $y\in \mathbb{R}^{\lvert y\rvert}$ are feed into and read from neuro-glial network after multiplied by  matrices $W_{\text{in}}^1, W_{\text{in}}^2$ and $W_{\text{out}}$.  Therefore,  the  network as a whole is represented by 

\begin{equation}\label{wholelearnnetwork}
\begin{aligned}
 &x_t=(1-\gamma)x_{t-1}+\gamma(W_{t-1}\phi(x_{t-1})+W_{\text{in}}^1 I_t)\\
 &W_t=(1-\gamma)W_{t-1}+\gamma (C\Phi(x_{t-1})+D\psi(z_{t-1}))\\
 &z_t=(1-\gamma \tau)z_{t-1}+\gamma \tau(F\psi(z_{t-1})+H\Phi(x_{t-1})+W_{\text{in}}^2 I_t)\\
 &y_t=W_{\text{out}}x_t +b_{\text{out}},
\end{aligned}
\end{equation}
where $b_{\text{out}}$ the bias vector  with the corresponding dimension.

\noindent
\textbf{Reinforcement learning procedure.}\\
We train instantiations of the discrete-time neuro-glial model to tackle the aforementioned tasks. The neuro-glial network architecture comprises $128$ neurons and $64$ astrocytes, with randomly initialized connections within each layer and interlayer hyperedges. The complete learning framework is depicted in Figure \ref{fig:main_result1}C. We first initialize the matrices $C$, $D$, $F$, $H$ in the way that the elements are drawn randomly from  normal distributions with zero mean, i.e.,
\[M_{ij}\sim \frac{1}{\sqrt{N_M}}\mathcal{N}(0,1),\]
where $N_M$ is the dimension of the focal matrix $M$.
The elements of input and output matrices $W_{\text{in}}^1$, $W_{\text{in}}^2$, $W_{\text{out}}$ and bias vector $ b_{\text{out}}$ are initialized from the uniform distribution $\mathcal{U}(-\frac{1}{\sqrt{N_M}},\frac{1}{\sqrt{N_M}})$, where $N_M$ is again the dimension.

The dimension of the output $y_t$ is the same as the number of actions in the bandits, i.e, $3$ in most simulations. After multiplied by the readout matrix and plus the bias, the output is fed to a softmax function, and it produces a probability distribution over the available actions $p_t=[p_{t}^1,p_{t}^2,p_{t}^3]$.   The probability of selecting the action $a_i \in \mathcal{A}$ is
\begin{equation}
    p_{t}^i=\frac{e^{y_{i}}}{\sum_1^3 e^{y_{j}}}, i=1,2,3.
\end{equation}
An action $a_t$ is then sampled from this probability distribution and subsequently executed by the agent. The bandit environment will provide the agent with a reward, represented as $r_{a_t}$.  
And according to \cite{rotman2023energy}, we use the loss function 
\[L=(\bar{r}_t-r_{a_t})\log p_{t}^i,\]
where $\bar{r}_t$ is the average of  rewards up to $t$ and $\log p_{t}^i$ is the logarithm of the probability.

At each trial, when the agent is presented with a new reward, the gradient of the loss function $L$ is calculated and used to update the network's parameters via the backpropagation (BP). During BP, we use the Adam method to optimize the aforementioned matrices and vectors with the default learning rate of $0.001$. 

In the case of other RNN-based methods as described below, we simply replace the neuro-glial network module with alternative network models. To ensure a fair comparison, all RNNs are constructed with $2$ stacked layers, with each layer consisting of $128$ units. The weights are initialized using the default method in PyTorch, and the training procedure remains consistent. 

The network architectures and training procedures are implemented using PyTorch in Python.

\noindent
\textbf{Learning performance comparison with different learning methods}
Numerous machine learning algorithms have been developed to tackle MABs. Among them, Upper Confidence Bound (UCB) and Thompson Sampling (TS) are widely recognized as the most prominent approaches for standard MABs.  Discounted UCB (DUCB) and switching-window UCB (SWUCB) have been devised to handle changing environments in non-stationary scenarios. In addition to these canonical bandit algorithms, some neuro-bandit algorithms that utilize feedforward or recurrent neural networks to model the agent's policy have been developed in recent years.

To perform a thorough yet not overly exhaustive assessment of learning performance, we analyze the asymptotic cumulative regret of our approach in comparison to selective algorithms across various scenarios. For stationary MABs, we evaluate our method against the UCB and TS algorithms, as well as RNN-based models including LSTM, vRNN, and GRU. In the context of non-stationary MABs, our method is compared to DUCB, SWUCB, and other RNN-based algorithms. It's worth noting that the training procedures for all RNN-based models remain consistent with the previously described methodology.

\bibliography{neuroglial.bib}


\bigskip 

\bmhead{Acknowledgments}
\noindent
We acknowledge support from ARO-MURI  grant W911NF2110312 from the US Department of Defense.

\bmhead{Authors' Contributions}
\noindent
L.G. and S.C. conceived of this project. L.G. performed analyses and simulations under the supervision of S.C..  T.P. and F.P. provided input on the problem formulation and interpretation of results. L.G. and S.C. drafted the manuscript. All authors edited the final manuscript.

\bmhead{Data Availability and Code Availability}
\noindent
All code for models, model learning, and analysis will be put online before the time of publication.

\bmhead{Competing Interests}
\noindent
The authors declare no competing interests.

\noindent
Correspondence and requests for materials can be addressed to L. Gong (glulu@wustl.edu).

\newpage

\titleformat{\section}
  {\normalfont\Large\bfseries}
  {SI}
  {1em}
  {}
\titleformat{\subsection}
  {\normalfont\bfseries}
  {SI.\arabic{subsection}}
  {1em}
  {}
\titleformat{\subsubsection}
  {\normalfont\bfseries}
{SI.\arabic{subsection}.\arabic{subsubsection}}
  {1em}
  {}
  
\section{Supplementary information}\label{secsi}
\subsection{Enhanced understanding of signal flow  in the tripartite synapse}\label{secsi-signalflow}
As stated in the main text, we can reveal the signal flow in the tripartite synapse more apparently via the symbolic description.
We denote the activities associated with the neurons and the astrocyte by symbols $V_N^1$, $V^{2}_N$,  and $V_A$ respectively, and the synaptic efficacy is $V_S$. Then, the interplay between astrocytes and neuronal elements can be represented by different arrows as shown in Figure \ref{fig:signal_flow}, where two instrumental feedback-loops are identified in the flow. In the first loop, the signal flows from Neuron 1 to Neuron 2 through the synapse, and Neurons 1 and 2's signals act on the synapse's efficacy (so-called Hebbian plasticity); Meanwhile,  Neurons 1 and 2's signals also affect the astrocyte's activity and in turn acts on the synaptic efficacy, and thus form the second loop based on the signal flow from Neuron 1 to Neuron 2. It is noticed that the  neurons 1 and 2's signals together act on the synapse and astrocyte. This integration of two signals is drastically different from two separated signals, and thus forms a high-order interaction.
\begin{figure}[htbp!]
    \centering
\includegraphics{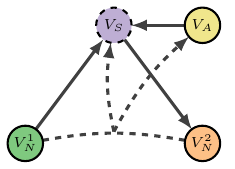}
    \caption{The signal flow in the tripartite synapse structure.}
    \label{fig:signal_flow}
\end{figure}

\subsection{Graphical description of neuro-glial population as a hypernetwork}\label{si_graphical_hypernetwork}
We have said that the neuro-glial population can be described by a two-layer hypernetwork, and now we will introduce a graphical description of this hypernetwork.

The brain is usually described by a network where  nodes represent  neurons and (directed) edges between nodes denote synaptic connections, which can be inhibitory or excitatory. In network theory, a generic (mono-layer) network can be  defined by a graph $G:=(N,E)$, where $N = \{1,...,n\} $ is the set of nodes; 
$E = \{(i,j)\lvert~ i,j\in N~ \rm{and~are~connected}\}$ is the set of edges. We consider the directional and weighted graph. That means an edge in $E$, e.g., $(i,j)$, has the direction from node $i$ to node $j$ and owns a weight $w_{ij}\in \mathbb{R}$.  In addition, a network can have multiple layers (neural and glial) that can have the same or different nodes \cite{Boccaletti2014}. 
When  considering a large number of neurons and astrocytes presenting in the brain system, we need to separate neurons and astrocytes as two groups because of their natural differences. In this regard, we prescribe that the neuro-glial populations have two different layers, with each layer representing the group of neurons and astrocytes respectively. And they are denoted by the graphs $G_n=\{N_n,E_n\}$ and $G_a=\{N_a,E_a\}$ respectively.

On the other hand, an edge only connects two nodes and thus describe the pairwise interaction. 
Because of the presence of high-order interaction within the tripartite synapse, the interconnections between the two layers cannot be represented by normal edges. We then extend it to hyperedges that can connect any number of nodes. 
A hyperedge $H_i$ is defined as a subset of $N$ satisfying $H_i\neq \emptyset$. 

Then, the whole neuro-glial network can be denoted by the hypernetwork $\mathcal{M}=\{G_n, G_a, \{H_i\}\}$.
This hypernetwork contains two layers, i.e, the neural layer and the glial layer. The neural layer consists of all the neurons and the intralayer network structure is consistent with the standard directional neural network; the glial layer includes the astrocytes and it embeds the directional network structure as well. Edges that connect a node to itself (self-loops) are excluded for both layers. If there are neurons interacting with an astrocyte, an interlayer connection takes place therein, and we use a hyperedge to represent it: each hyperedge (the triangle shape in Figure \ref{fig:tripartitesynapse}B) connects one astrocyte and two neurons. Noting that the hyperedge essentially includes the edges connecting the two neurons which represents the plastic synapses. In this sense, the defined hyperedge can well capture the high-order interactions between two neurons, the astrocyte and synapse. As a result, the neuro-glial ensemble structure is well represented by this two-layer hypernetwork.

\subsection{Well-definiteness of the neuro-glial network model}\label{secsi1}
From the mathematical point of view, it is important to check if the neuro-glia network model is posed and defined correctly.  
Our  model is given by a set of continuous-time ODEs. For such a system,  it is well-posed if the solution to an initial value problem exists and is unique. The well-posedness is guaranteed if the vector field of the model is Lipschitz continuous in the variables and continuous in time \cite{perko2013differential}. It is well-known that  common activation functions, such as the logistic sigmoid and $tanh$ are Lipshitz. In the presumption that the external inputs are continuous in time, our model is well-posed.

The model is also well-defined in the sense that the dynamics will not expand without restriction but will be confined in an appropriate subspace in the real space for any initial conditions after some certain time  as shown in \hyperref[secsi2]{SI.3}. Moreover, we can approximately estimate this region using mathematical arguments as shown in the following.

\subsection{Normal analysis of network motif dynamics}\label{secsi2}
As the neuro-glial model in this work is proposed for the first time, it is necessary to conduct a normal analysis including studying the boundedness and fixed points conditions. The network is composed of the network motifs, it is enough to consider the minimal model on the network motifs. 

\noindent
\textbf{Boundedness}
Note the network motif dynamics are given by
\begin{equation}\label{networkmotifdyanmics_si}
 \begin{aligned}
 &\tau_1 \dot{x}_1=-a_1x_1+w_2\phi(x_2)+u_1(t)\\
 &\tau_1 \dot{x}_2=-a_2x_2+w_1\phi(x_1)+u_2(t)\\
 &\tau_2 \dot{w}_{1}=-b_{1}w_{1}+c_{1}\phi(x_1)\phi(x_2)+d_{1}\psi(z)\\
 &\tau_2\dot{w}_{2}=-b_{2}w_{2}+c_{2}\phi(x_1)\phi(x_2)+d_{2}\psi(z)\\
 &\tau_3 \dot{z}=-ez+h\phi(x_1)\phi(x_2)+v(t).
 \end{aligned} 
\end{equation}
We first define the boundedness of  a general dynamical system without inputs. 
\begin{definition}
Given a dynamical system 
\[\dot{x}=f(x),\]
where $f:\mathbb{R}^n\rightarrow \mathbb{R}^n$ is continuous. Let $x(t), ~t\geq 0$ be a trajectory of the system. $x(t)$ is said to be \emph{ultimately bounded} if there exist $M>0$ and $T>0$ such that $\lVert x(t)\rVert \leq M$ for all $t\geq T$. Moreover,  the system is said to be ultimately bounded if all trajectories are ultimately bounded.
\end{definition}

For system \eqref{networkmotifdyanmics_si} without external inputs, although the vector field are defined for $\mathbb{R}^5$, we can show that the dynamics are indeed bounded after some certain time. Let $X=(x_1,x_2,x_3,w_{1},w_{2},z)^\top$. As the  activation functions are bounded, we denote the maximum values of   $\phi(\cdot)$ and $\psi(\cdot)$ by $M_1>0$ and $M_2>0$ respectively. Define the set 
\[\Omega:=\left \{X\in \mathbb{R}^5: \begin{aligned}
&\lvert x_1\lvert,\lvert x_2\lvert\leq x_{\rm{max}}\\
&\lvert w_1\lvert,\lvert w_2\lvert\leq w_{\rm{max}}\\
&\lvert z\lvert \leq z_{\rm{max}}
\end{aligned} \right\},\]
with $x_{\rm{max}}=w_{\rm{max}}M_1/\min\{a_1,a_2\}$, $w_{\rm{max}}=(\max\{\lvert c_1\lvert,\lvert c_2 \lvert\}M^2_1+\max\{\lvert d_1\lvert,\lvert d_2 \lvert\}M_2 /\min\{b_1,b_2\}$ and $z_{\rm{max}}=\lvert h\lvert M^2_1/e$.

\begin{theorem}\label{boundedsystem}
In absence of external inputs, the dynamics of \eqref{networkmotifdyanmics_si} are ultimately bounded in the set $\Omega$.
\end{theorem}
\begin{proof}
 Let $\bar{z}(t)$ be the solution of the differential equation 
\[\dot{\bar{z}}=-e\bar{z}+\lvert h\lvert M_1^2,~~ \bar{z}(0)=z(0). \]
One can obtain that $\bar{z}(t)=(z(0)-\lvert h\lvert M_1^2/e)\exp(-et)+\lvert h\lvert M_1^2/e,$ which yeilds
\[\lvert \bar{z}(t)\lvert \leq \lvert (z(0)-\lvert h\lvert M_1^2/e)\exp(-et)\lvert + h\lvert M_1^2/e.\]
As $e>0$, the first term on the right hand side of the above inequality converges to zero exponentially. Therefore, there exist a time $T>0$ such that $\lvert \bar{z}(t)\lvert\leq h\lvert M_1^2/e$ for all $t>T$. On the other hand, from the last equation of \eqref{networkmotifdyanmics_si}, we have $\dot{z}\leq \dot{\bar{z}}$. Then, by the comparison lemma \cite[Lemma 3.4]{khalil1996nonlinear}, we get $\lvert z(t)\lvert\leq \lvert \bar{z}(t)\lvert$. It follows that there exists a time $T_1$ such that $\lvert z(t)\lvert\leq h\lvert M_1^2/e$ for all $t>T_1$. 

Analogously, we can derive the bound for the other variables as defined in $\Omega$. Finally, we prove that the dynamics of \eqref{networkmotifdyanmics_si} are ultimately bounded in the set $\Omega$.
\end{proof}

Along with the Theorem of boundedness, we have some remarks.
\begin{remark}
\begin{enumerate}[1.]
    \item Here, we focus on the autonomous system. With the proper condition that the inputs are bounded, one can  show the boundedness of the system when external inputs are presented \cite{9993009}. 
    \item The set $\Omega$ is positive invariant and attractive with respect to \eqref{networkmotifdyanmics_si}. Intuitively, by the definition of limit points \cite{wiggins1990introduction}, all the positive limit points of system \eqref{networkmotifdyanmics_si}, such as fixed points and limit cycles, must be included in the attractive set $\Omega$. This property is helpful  for obtaining the following results about fixed points.  
\end{enumerate}
\end{remark}
\noindent
\textbf{Fixed points} With the boundedness in hand,  we can study the fixed points of the system.
Continue considering the system \eqref{networkmotifdyanmics_si} without external inputs.  Letting the right hand sides of \eqref{networkmotifdyanmics_si}  be zero results in the following equations
\begin{equation}\label{nullclines}
    \begin{aligned}
&x_1=\frac{w_2\phi(x_2)}{a_1} \\ 
    &x_2=\frac{w_1\phi(x_1)}{a_2} \\
&w_1=\frac{c_1\phi(x_1)\phi(x_2)+d_1\psi(z)}{b_1}\\
&w_2=\frac{c_2\phi(x_1)\phi(x_2)+d_2\psi(z)}{b_2}\\
    &z=\frac{h\phi(x_1)\phi(x_2)}{e}.
    \end{aligned}
\end{equation}
Each of the above equations defines a nullcline in $\mathbb{R}^5$, and together their solutions (intersections of nullclines) yield the fixed points.

First, let us consider the existence of fixed points, which can be proved easily by using the Brouwer's fixed point theorem \cite{karamardian2014fixed}. Define the mapping 
\[F:=\left( \begin{aligned}
    &\frac{w_2\phi(x_2)}{a_1}\\
    &\frac{w_1\phi(x_1)}{a_2}\\
    &\frac{c_1\phi(x_1)\phi(x_2)+d_1\psi(z)}{b_1}\\
    &\frac{c_2\phi(x_1)\phi(x_2)+d_2\psi(z)}{b_2}\\
    &\frac{h\phi(x_1)\phi(x_2)}{e}
\end{aligned}\right).\]

\begin{theorem}
System \eqref{networkmotifdyanmics_si}  has (at least) a fixed point in $\Omega$.
\end{theorem}
\begin{proof}
It is easy to check that the defined mapping $F$ is continuous. To prove the existence of fixed points,
according to Brouwer's Fixed-point Theorem, we only need to show the set $\Omega$ is compact and convex. Since $\Omega$ is bounded and closed, the compactness follows. In addition, the set $\Omega$ actually defines a hyper rectangle in $\mathbb{R}^5$, which is convex. Therefore, we can conclude that there exists (at least) one fixed point of system \eqref{networkmotifdyanmics_si} in $\Omega$.
\end{proof}

Next,  we examine the uniqueness of the fixed point in \eqref{networkmotifdyanmics_si}. The Jacobian of $F$ is given by 
\begin{equation}
\begin{aligned}
    \jac{F}&= \begin{bmatrix}
    0&\frac{w_2\phi'(x_2)}{a_1}&0&\frac{\phi(x_2)}{a_1}&0\\
    \frac{w_1\phi'(x_1)}{a_2}&0&\frac{\phi(x_1)}{a_2}&0&0\\
    \frac{c_1\phi'(x_1)\phi(x_2)}{b_1}&\frac{c_1\phi(x_1)\phi'(x_2)}{b_1}&0&0&\frac{d_1\psi'(z)}{b_1}\\
    \frac{c_2\phi'(x_1)\phi(x_2)}{b_2}&\frac{c_2\phi(x_1)\phi'(x_2)}{b_2}&0&0&\frac{d_2\psi'(z)}{b_2}\\
    \frac{h\phi'(x_1)\phi(x_2)}{e}&\frac{h\phi(x_1)\phi'(x_2)}{e}&0&0&0
    \end{bmatrix}
    \end{aligned}
\end{equation}
To ensure that Eqs. \eqref{nullclines} have a unique solution in the previously obtained set $\Omega$, one sufficient condition is that the 
inequality
\begin{equation}\label{1normJacobian}
    \lVert \jac{F}\rVert_{\Omega}=\sup_{X\in \Omega} \lVert \jac{F}(X)\rVert <1
\end{equation}
holds, where $\lVert \cdot \rVert$ denotes the matrix norm.
We take the $1$-norm of $\jac{F}$, i.e., the maximum of the  absolute values sum of the rows
\begin{equation}\label{1normJacobian1}
\begin{aligned}
    &\lVert \jac{F}(X)\rVert=\\
    &\max \left\{ 
     \left\lvert\frac{w_2\phi'(x_2)}{a_1} \right \rvert+ \left\lvert\frac{c_1\phi(x_1)\phi'(x_2)}{b_1} \right \rvert+ \left\lvert\frac{c_2\phi(x_1)\phi'(x_2)}{b_2} \right \rvert +\left\lvert\frac{h\phi(x_1)\phi'(x_2)}{e}\right\rvert, \right.\\ 
    &\quad\quad\quad\quad\left. \left\lvert\frac{w_1\phi'(x_1)}{a_2} \right \rvert+ \left\lvert\frac{c_1\phi'(x_1)\phi(x_2)}{b_1} \right \rvert+ \left\lvert\frac{c_2\phi'(x_1)\phi(x_2)}{b_2} \right \rvert+\left\lvert\frac{h\phi'(x_1)\phi(x_2)}{e}\right\rvert,\right.\\ &\quad\quad\quad\quad\left. \left\lvert\frac{\phi(x_1)}{a_2} \right \rvert, \left\lvert\frac{\phi(x_2)}{a_1} \right \rvert, \left\lvert\frac{d_1\psi'(z)}{b_1}\right\rvert + \left\lvert\frac{d_2\psi'(z)}{b_2}\right\rvert\right\}.
    \end{aligned}
\end{equation}
Note that all the variables and the derivatives of activation functions are bounded. It is always possible to find such conditions that \eqref{1normJacobian1} is less than $1$ for all points in $\Omega$. To showcase, let us consider the case where $\phi(\cdot)=1/(1+e^{-x})$, $\phi'(\cdot)=e^{-x}/(1+e^{-x})^2$ and $\psi(z)=(e^{z}+e^{-z})/(e^{z}+e^{-z})$, $\psi'(z)=1-\psi^2(z)$. Then we have 
\[\sup \left\lvert\frac{\phi(x_1)}{a_2} \right \rvert=\frac{1}{a_2},\]
\[ \sup \left\lvert\frac{\phi(x_2)}{a_1} \right \rvert=\frac{1}{a_1},\]
\[\sup \left(\left\lvert\frac{d_1\psi'(z)}{b_1}\right\rvert + \left\lvert\frac{d_2\psi'(z)}{b_2}\right\rvert\right)=\left\lvert\frac{d_1}{b_1}\right\rvert+\left\lvert\frac{d_2}{b_2}\right\rvert,\]
\begin{equation*}
    \begin{aligned}
    & \sup \left(\left\lvert\frac{w_1\phi'(x_1)}{a_2} \right \rvert+ \left\lvert\frac{c_1\phi'(x_1)\phi(x_2)}{b_1} \right \rvert+ \left\lvert\frac{c_2\phi'(x_1)\phi(x_2)}{b_2} \right \rvert +\left\lvert\frac{h\phi'(x_1)\phi(x_2)}{e}\right\rvert\right)\\
    &=\frac{\lvert w_1\rvert_{\max} }{4a_2}+\left\lvert\frac{c_1}{4b_1}\right \rvert+\left\lvert\frac{c_2}{4b_2}\right \rvert+\left\lvert\frac{h}{4e}\right \rvert,\\
    &\sup \left(\left\lvert\frac{w_2\phi'(x_2)}{a_1} \right \rvert+ \left\lvert\frac{c_1\phi(x_1)\phi'(x_2)}{b_1} \right \rvert+ \left\lvert\frac{c_2\phi(x_1)\phi'(x_2)}{b_2} \right \rvert +\left\lvert\frac{h\phi(x_1)\phi'(x_2)}{e}\right\rvert\right)\\
    &=\frac{\lvert w_2\rvert_{\max} }{4a_1}+\left\lvert\frac{c_1}{4b_1}\right \rvert+\left\lvert\frac{c_2}{4b_2}\right \rvert+\left\lvert\frac{h}{4e}\right \rvert,
    \end{aligned}
\end{equation*}
where $\lvert w_i\rvert_{\max}, i=1,2$ are the maximum values of $\lvert w_i \rvert$ taking in $\Omega$.
Therefore, one has $\lVert \jac{F}\rVert<1$ for all $X\in \Omega$ if the following conditions are satisfied
\begin{equation}\label{conditionofunquefixed point}
  a_1>1,~a_2>1,~\left\lvert\frac{d_1}{b_1}\right\rvert+\left\lvert\frac{d_2}{b_2}\right\rvert<1,~\max \left\{\frac{\lvert w_2\rvert_{\max} }{a_1},\frac{\lvert w_1\rvert_{\max} }{a_2} \right\}+\left\lvert\frac{c_1}{b_1}\right \rvert+\left\lvert\frac{c_2}{b_2}  \right \rvert +\left\lvert\frac{h}{e}  \right \rvert <4.
\end{equation}
Then, it follows that the mapping $F$ is a contraction mapping in $\Omega$ \cite{chicone2008ordinary}. A contraction mapping has the property for admitting a unique fixed point as stated in Banach's fixed point theorem \cite{chicone2008ordinary}. 
According to this theorem,  $F$ has a unique fixed point  in $\Omega$ , which implies the system \eqref{networkmotifdyanmics_si} has a unique fixed point as stated in the following theorem.
\begin{theorem}\label{uniquefixed point}
When \eqref{conditionofunquefixed point} holds, system \eqref{networkmotifdyanmics_si} has a unique fixed point in the defined domain, and this fixed point is located in the set $\Omega$.
\end{theorem}
\begin{remark}
In the above process, we used the $1$-norm of $\jac{F}$ and arrive at the sufficient conditions \eqref{conditionofunquefixed point}. Of course, one can use other norms and thus obtain different sufficient conditions for the uniqueness of the fixed point.
\end{remark}

Next, to go beyond the single fixed point, we investigate the conditions for the existence of multiple fixed points. 
As the involvement of so many parameters and uncertain activation functions, it is difficult to fully and analytically characterize the parameter conditions for the existence of multiple fixed points. For simplicity, we restrict to the case of sigmoid and hyperbolic tangent activation functions, i.e., $\phi(x)=1/(1+e^{-x})$ and $\psi(z)=(e^{z}+e^{-z})/(e^{z}+e^{-z})$.

In \eqref{nullclines}, we substitute the third and fourth equations to the first two, and arrive at the following set of equations with the reduced dimension
\begin{equation}\label{reducednullclines}
    \begin{aligned}
    &x_1=\frac{(c_2\phi(x_1)\phi(x_2)+d_2\psi(z))\phi(x_2)}{a_1b_2} \\ 
    &x_2=\frac{(c_1\phi(x_1)\phi(x_2)+d_1\psi(z))\phi(x_1)}{a_2b_1} \\
    &z=\frac{h\phi(x_1)\phi(x_2)}{e}.
    \end{aligned}
\end{equation}
Now we presume that the values of variables $x_1$, $x_2$ and $z$ are large  so that $\phi(x_1)=1/(1+e^{-x_1})= 1-\sigma_1$, $\phi(x_2)=1/(1+e^{-x_2})= 1-\sigma_2$, and $\psi(z)=(e^{z}+e^{-z})/(e^{z}+e^{-z})= 1-\sigma_3$ where $0<\sigma_1,\sigma_2,\sigma_3<1$ are   small. In doing so, \eqref{reducednullclines} yields
\begin{equation}\label{reducednullcline}
    \begin{aligned}
    &x_1^*=\frac{c_2(1-\sigma_1)(1-\sigma_2)^2+d_2(1-\sigma_2)(1-\sigma_3)}{a_1b_2} \\ 
    &x_2^*=\frac{c_1(1-\sigma_1)^2(1-\sigma_2)+d_1(1-\sigma_1)(1-\sigma_3)}{a_2b_1}\\
    &z^*=\frac{h(1-\sigma_1)(1-\sigma_2)}{e}.
    \end{aligned}
\end{equation}
To ensure that the solution given by  \eqref{reducednullcline} is one fixed point of system \eqref{networkmotifdyanmics_si}, the following equalities should also be satisfied
\begin{equation}\label{reducednullclines1}
    \begin{aligned}
    &1/(1+e^{-x_1^*})= 1-\sigma_1 \\ 
    &1/(1+e^{-x_2^*})= 1-\sigma_2\\
    &(e^{z^*}-e^{-z^*})/(e^{z^*}+e^{-z^*})= 1-\sigma_3.
    \end{aligned}
\end{equation}
Now, the question turns to be finding a collection of parameters $a_1,...,h$ such that \eqref{reducednullclines1} holds with $\sigma_1,\sigma_2,\sigma_3>0$ being very small.  We further simplify this problem by assuming that $a_1b_2=a_2b_1,c_1=c_2,d_1=d_2$ and $\sigma_1=\sigma_2=\sigma_3$. Then we have  
\begin{equation}\label{reducednullclines2}
1+e^{-x_1^*}=1+e^{-x_2^*}=(e^{z^*}+e^{-z^*})/(e^{z^*}-e^{-z^*})= 1/(1-\sigma).
\end{equation}
Substituting \eqref{reducednullcline}  results in
\begin{equation*}
    \begin{aligned}
    &\frac{c_2(1-\sigma)^3+d_2(1-\sigma)^2}{a_1b_2}=\ln{\frac{\sigma}{1-\sigma}} \\ 
    &\frac{2h(1-\sigma)^2}{e}=\ln{\frac{2-\sigma}{\sigma}}.
    \end{aligned}
\end{equation*}
It then yields
\begin{equation}\label{reducednullclines3}
  \frac{ec_2(1-\sigma)^3+(ed_2+2ha_1b_2)(1-\sigma)^2}{a_1b_2e} =  \ln{\frac{2-\sigma}{1-\sigma}}
\end{equation}
 The right hand side of \eqref{reducednullclines3} is in the range $\mathbb{R}_+$ and is monotonically increasing  for $\sigma\in (0,1)$. Note that $a_1$, $b_2$, $c_2$, $d_2$, $e$, $h$ are free parameters. When they are all positive, the left hand side of \eqref{reducednullclines3}
 is in the range $(0, \frac{ec_2+ed_2+2ha_1b_2}{a_1b_2e})$ and is monotonically decreasing for $\sigma\in (0,1)$. That means the two sides must have one intersection for $\sigma\in (0,1)$. In this case, there will be a unique fixed point. 
 On the other hand, such many free parameters can give rise to other possibilities. We can fix some parameters to be constant values, e.g., $a_1b_2=0.1$, $d_2=e=1$, $h=1$.   It can be calculated that when $-1.5595<c_2<-1.1307$,  the two sides  of \eqref{reducednullclines3} will always have intersections as shown in Figure \ref{fig:two_intersections}B, which results in two fixed points for the system. 
 It can be expected that when we release these restrictions on the parameters, it is easier for the system to have more fixed points.  
 \begin{figure}
     \centering
\includegraphics[width=10cm]{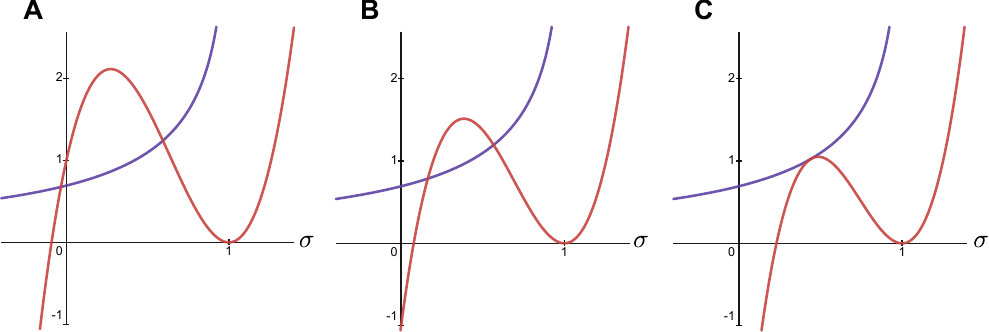}
          \caption{Intersections of two curves where the blue line is the right hand side and brown line is the left hand side. A.  one intersection when $c_2=-1.1$; B. two intersections when $c=-1.3$; C one intersection when $c=-1.5595$.}
     \label{fig:two_intersections}
 \end{figure}

\begin{remark}
We have shown in the above how the system can admit two fixed points analytically and numerically. The case of multiple fixed points is not rare because of the many parameters, and an example of $3$ fixed points can be obtained under some conditions as in Figure \ref{fig:networkmotif}. To show the case of more fixed points is tedious and marginal, an thus it is out of scope of this work. 
\end{remark}

\begin{figure}[t!]
    \centering
\includegraphics[width=11cm]{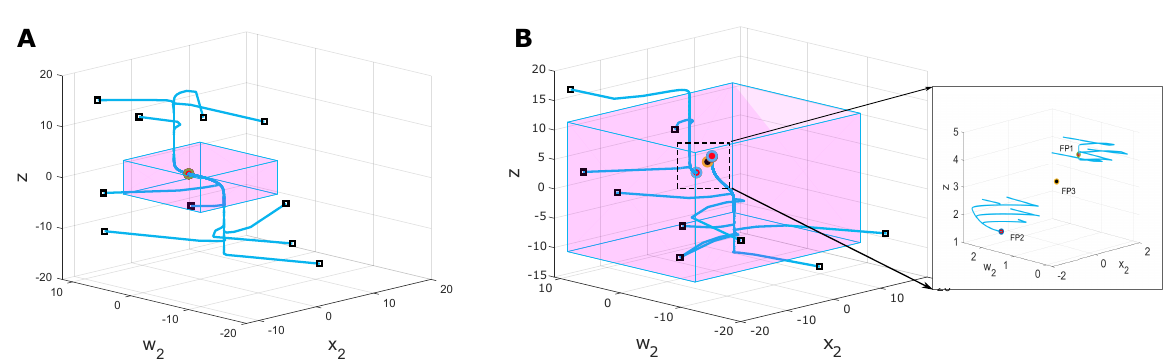}
    \caption{The cases of a single and multiple fixed points: pink cubes are the sets $\Omega$; red dots are stable fixed points while the black dot is the unstable one. The time-scale parameters are $\tau_1=\tau_2=0.01$, $\tau_3=1$, and in A, the other parameters are $a_1=2$,
$a_2=1$,
$b_1=1.2$,
$b_2=1.7$,
$c_1=2$,
$c_2=-3$,
$d_1=-4$,
$d_2=5$,
$e=2$, $h=6.6$, while in B $a_1=0.7$,
$a_2=0.6$,
$b_1=1.6$,
$b_2=1.7$,
$c_1=12$,
$c_2=-10$,
$d_1=-4$,
$d_2=5$,
$e=0.6$, $h=6$.}
\label{fig:single_multiple_fixed_points}
\end{figure}

\subsection{Singular perturbation analysis}\label{secsi3}
By a change of time  $t'=\epsilon t$, we can rewrite the network motif dynamics without external inputs into  
\begin{equation}\label{networkmotifdynamics_si1}
 \begin{aligned}
 &x'_1=-a_1x_1+w_2\phi(x_2)\\
 &x'_2=-a_2x_2+w_1\phi(x_1)\\
 &w'_{1}=-b_{1}w_{1}+c_{1}\phi(x_1)\phi(x_2)+d_{1}\psi(z)\\
 &w'_{2}=-b_{2}w_{2}+c_{2}\phi(x_1)\phi(x_2)+d_{2}\psi(z)\\
 &z'=\epsilon(-ez+h\phi(x_1)\phi(x_2)),
 \end{aligned} 
\end{equation}
where $'=d/d t'$ denotes the differentiation with respect to $t'$. Because of the nature of small value of $\epsilon$, \eqref{networkmotifdyanmics_si} (or \eqref{networkmotifdynamics_si1}) indeed defines a perturbation problem. The \emph{singular perturbation theory} \cite{kuehn2015multiple} has been developed to solve such problems in the past few decades. In the following, we turn to the analysis of system \eqref{networkmotifdyanmics_si} from the singular perturbation perspective. 

By setting $\epsilon=0$ in \eqref{networkmotifdynamics_si1}, we obtain the singular limit of the system, i.e.,
\begin{equation}\label{networkmotifdynamics2}
 \begin{aligned}
 &x'_1=-a_1x_1+w_2\phi(x_2)\\
 &x'_2=-a_2x_2+w_1\phi(x_1)\\
 &w'_{1}=-b_{1}w_{1}+c_{1}\phi(x_1)\phi(x_2)+d_{1}\psi(z)\\
 &w'_{2}=-b_{2}w_{2}+c_{2}\phi(x_1)\phi(x_2)+d_{2}\psi(z)\\
 &z'=0.
 \end{aligned} 
\end{equation}

In the above system, the derivative of  $z$ is  zero. Intuitively, one can consider that the $z$-variable is fixed as in initial conditions, i.e., $z(t)=z_0 \in \mathbb{R}$. It results in the flowing truncated system
\begin{equation}\label{fastsubsystemnetworkmotifdynamics}
 \begin{aligned}
 &x'_1=-a_1x_1+w_{2}\phi(x_2)\\
 &x'_2=-a_2x_2+w_{1}\phi(x_1)\\
 &w'_{1}=-b_{1}w_{1}+c_{1}\phi(x_1)\phi(x_2)+d_{1}\psi(z_0)\\
 &w'_{2}=-b_{2}w_{2}+c_{2}\phi(x_1)\phi(x_2)+d_{2}\psi(z_0)\\
 \end{aligned} 
\end{equation}
The above system captures the dynamics on the neuronal layer, i.e., coupled rate-based RNN and Hebbian learning of synapses. Since $z_0$ is a constant in system \eqref{fastsubsystemnetworkmotifdynamics}, we take it as a non-dynamic parameter that can take different values. 

As  shown  in Theorem \ref{boundedsystem}, system \eqref{networkmotifdyanmics_si} is bounded . Let $z_{\min}$ and $z_{\min}$ represent the minimum and maximum values that variable $z$ can take in $\Omega$. Under the assumption that $z(t)$ can span the whole admitted space in $\Omega$, we have that $z_0\in [z_{\min}, z_{\max}]$. 
As a consequence, if the dynamics of the neuronal layer exhibit critical changes, such as changes of the number and/or stability of the fixed points,  we can say there exist bifurcations in \eqref{fastsubsystemnetworkmotifdynamics} with respect to $z_0$, and these bifurcations are indeed induced by the  self-slowly-varying astrocytic process.

In the following part, we will analyze the dynamics of the subsystems \eqref{fastsubsystemnetworkmotifdynamics} in the spirit of the above idea.

\noindent
\textbf{Astrocytes regulate  neural dynamics}
We visualize the change process of the fixed point set of system \eqref{fastsubsystemnetworkmotifdynamics} as a consequence of the perturbation inducing from the constant glial signal. Since there are $4$ variables in Eqs. \eqref{nullclines}, the first step will be reducing the dimension, otherwise it is difficult to visualize these nullclines in 3D coordinates. By eliminating the variable $w_2$, we have
\begin{subequations}\label{nullclines2}
    \begin{align}
    \label{nullclines2a}
    &x_2=\frac{w_1\phi(x_1)}{a_2} \\
    \label{nullclines2b}&w_1=\frac{c_1\phi(x_1)\phi(x_2)+d_1\psi(z_0)}{b_1}\\
    \label{nullclines2c}&\frac{a_1x_1}{\phi(x_2)}=\frac{c_2\phi(x_1)\phi(x_2)+d_2\psi(z_0)}{b_2}.
    \end{align}
\end{subequations}
Note that the activation function  $\psi_{\min}\leq \psi(z_0)\leq \psi_{\max}$. To examining the change of fixed points is equivalent to studying the change of the intersection of \eqref{nullclines2a}-\eqref{nullclines2c} as $\psi(z_0)$ varies in $[\psi_{\min},\psi_{\max}]$. Recall that each equation of \eqref{nullclines2} defines a manifold in $(x_1,x_2,w_1)\in \tilde{\Omega}_{\rm motif} \subset \mathbb{R}^3$. We can show these manifolds geometrically for given parameters, where Figure \ref{fig:intersectionchanging} displays the situations for $\psi(z_0)=\psi_{\min}$, $\psi(z_0)=0$, and $\psi(z_0)=\psi_{\max}$ respectively under the parameter condition
$a_1=0.3,a_2=0.4,b_1=1,b_2=0.5,c_1=6,c_2=-5,d_1=-2,d_2=3$.

\begin{figure}[htbp!]
    \centering
    	\includegraphics[width=10cm]{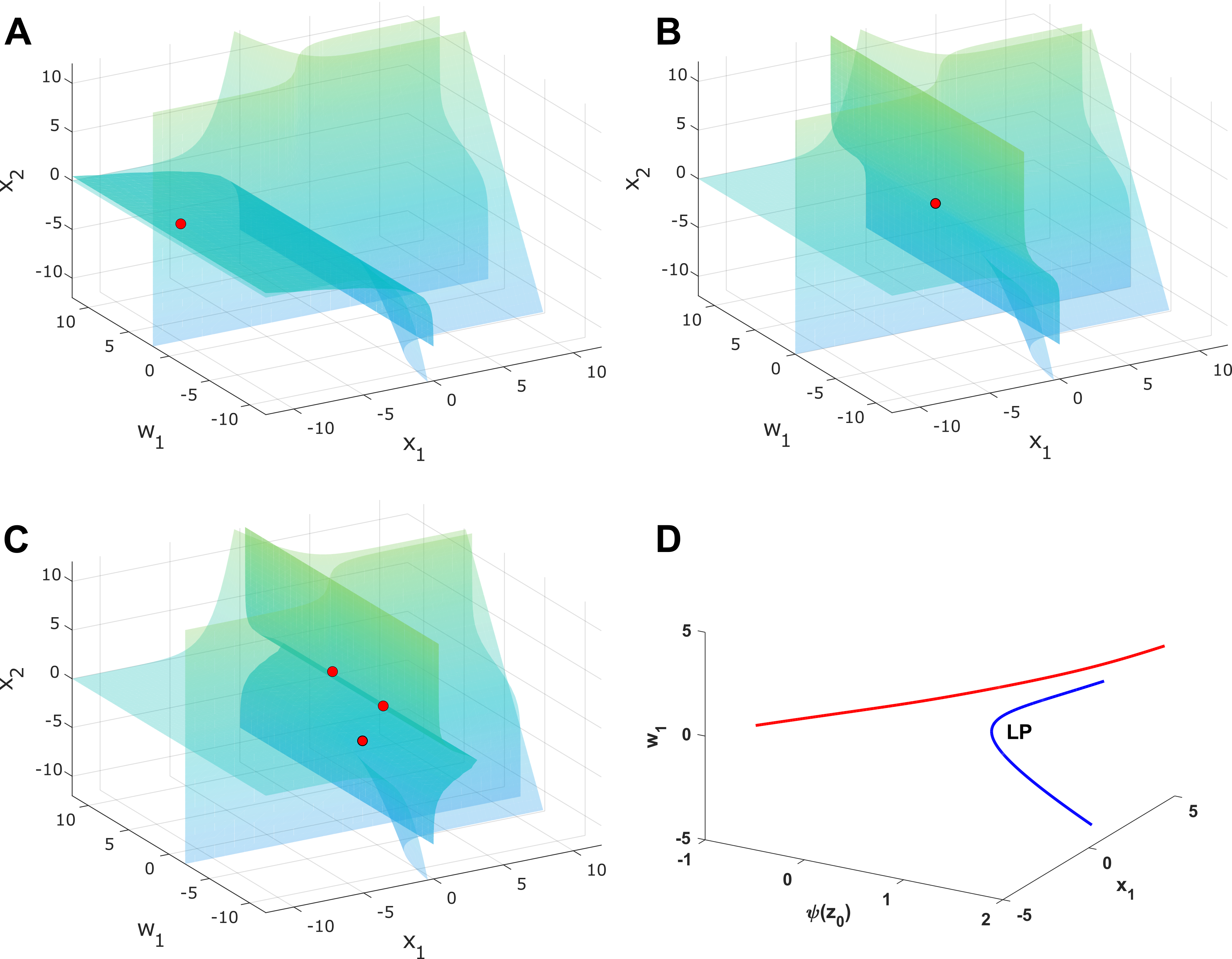}
\caption{The fixed points (red dots) of the neural subsystem for different values of $\psi(z_0)$: A. $\psi(z_0)\approx -1$, B. $\psi(z_0)=0$, C. $\psi(z_0)\approx 1$   In D, LP denotes the bifurcation point.} 
\label{fig:intersectionchanging}
\end{figure}

In Figure \ref{fig:intersectionchanging}, when $\psi(z_0)\approx -1$ there is one fixed point (red dot); as $\psi(z_0)$ increases, the position of this fixed point changes accordingly. When $\psi(z_0)\approx1$, another $2$ fixed points exist. This means there is an increase of fixed point points  at a certain value of $\psi(z_0)$, and this is confirmed by the bifurcation diagram obtained with Matcont ( see Fig \ref{fig:intersectionchanging}D). It indicates that a branch of fixed points (red line) always exists. In contrast,  the other branch of fixed points (blue line) exists when $\psi(z_0)\geq 0.7818$, but a saddle-node bifurcation occurs at $\psi(z_0)= 0.7818$ such that these two fixed points collide and annihilate each other. We call this process a \emph{pseudo-bifurcation} resulted from the change of the glial activity. And the neural dynamic behaviors are regulated by the glial process in this top-down manner.

\subsection{Extended simulation results}\label{secsi4}
In this section, we provide extra simulations  that are complimentary to the results in the main text.
\subsubsection{Regrets and converging time}\label{si_regrets_time}
Figure \ref{fig:extra_result_regrets} shows the details of each method in the learning procedure. As shown in the plots, the regrets of every method are dense at the beginning because the agent needs to explore the environment. After enough time, the agent can make the optimal actions such that there are no more regrets. It is observed that neuro-glial method does not generate regrets after about $2000$ trials  while other methods still give rise to regrets in the remaining trials. Therefore, the neuro-glial method takes the shortest time to converge. In the right plot, we can see that the neuro-glial method accounts for a medium amount of asymptotic cumulative regret among all the methods, while TS method has the lowest and UCB method has the highest asymptotic cumulative regret.
\begin{figure}[htbp!]
    \centering
\includegraphics[width=11cm]{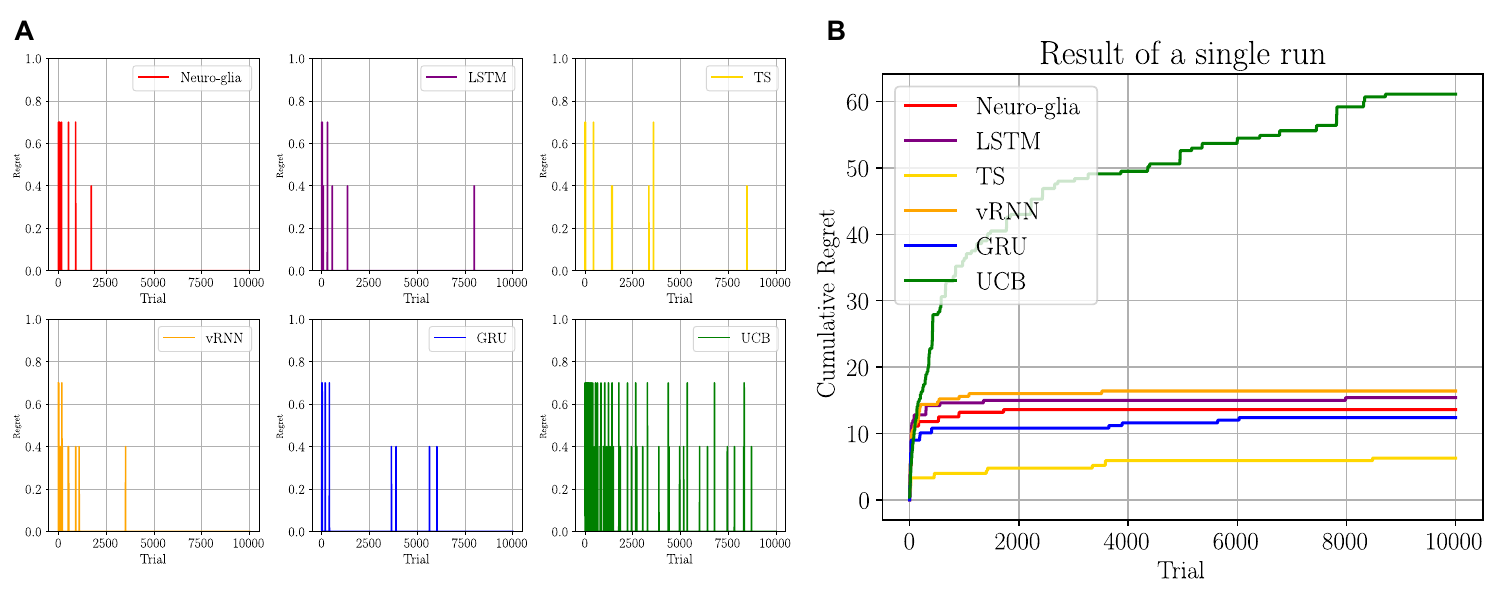}
    \caption{Regrets per trial and cumulative regrets of each method. A shows the regrets per trial for different methods, while B shows the cumulative regrets over trials.}
    \label{fig:extra_result_regrets}
\end{figure}

\subsubsection{Robustness analysis in diverse stationary bandits}\label{secsirobust} 
We conduct a robustness analysis of all methods under various conditions of arm probabilities, specifically $(\mu_1, \mu_2, \mu_3) = (0.6-\lambda, 0.6, 0.6+\lambda)$ with $0 < \lambda < 0.4$. The UCB method is less competitive and excluded from this comparison due to its significantly larger regrets. As $\lambda$ decreases, the bandit becomes more challenging due to the arms' probabilities converging. To vary the difficulty of the bandit tasks, we change $\lambda$ from $0.38$ to $0.02$ 
(see Figure \ref{fig:result_robust}). It is recognized that neuro-glial method performs similarly to other RNN-based methods under larger $\lambda$ values. However, our method tends to exhibit better and more robust performance in more challenging bandits in contrast to others, which experience a decline in performance. 

\begin{figure}[htbp!]
    \centering
\includegraphics[width=11cm]{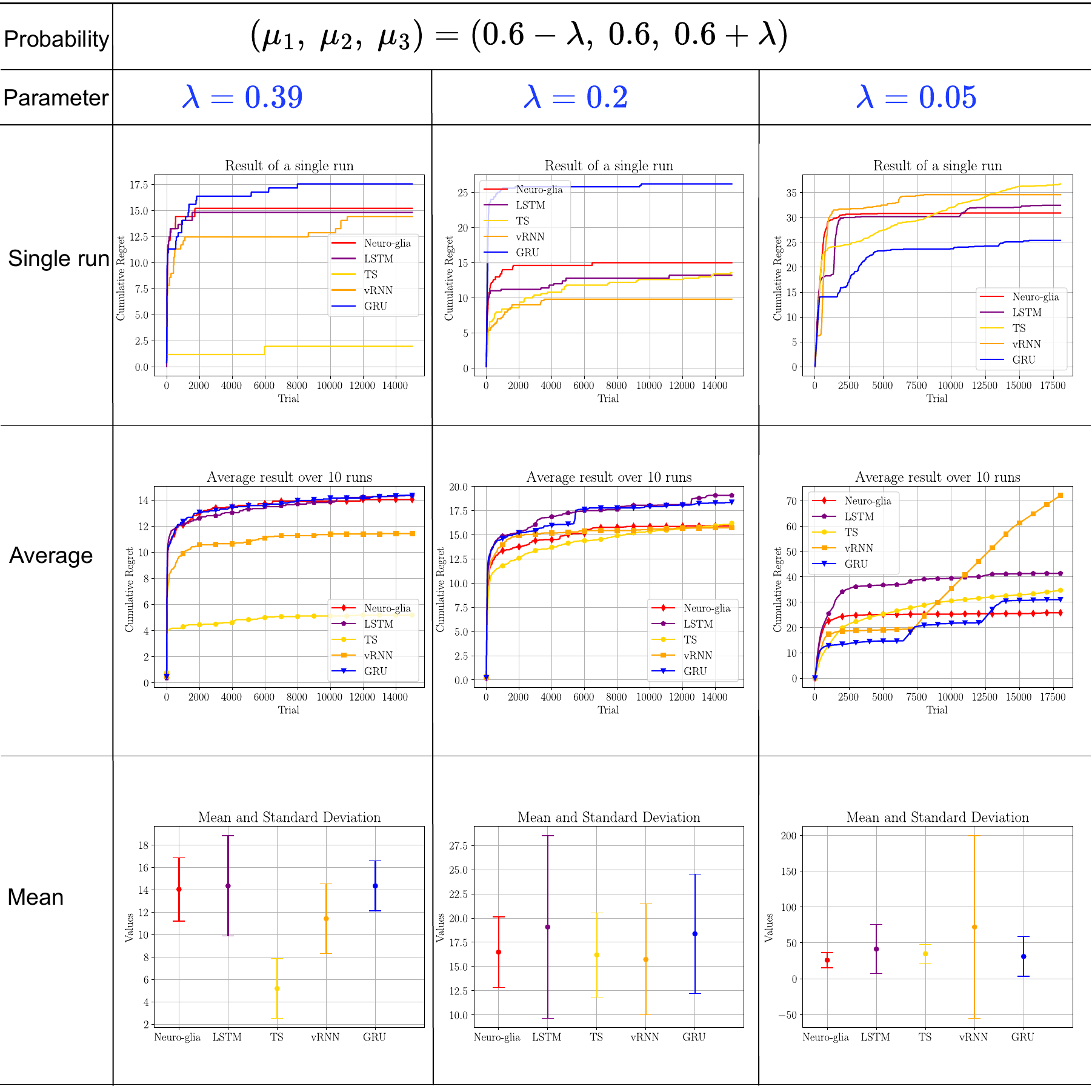}
    \caption{The learning performance robustness of each method  is examined under different conditions of arm probabilities.}
    \label{fig:result_robust}
\end{figure}

\subsubsection{Detailed performance comparison in non-stationary bandits}\label{si-detail_non-stationary}
Figure \ref{fig:extra_result_nonstationary_full} gives a comprehensive comparison of different methods in the non-stationary bandits. The neuro-glial method consistently attains the lowest and maintains near-stationary asymptotic cumulative regret, both in single and multiple runs. In contrast, other methods struggle to adapt to evolving environments, resulting in steadily increasing regrets. This result is corroborated in experiments of both the flip-flop and smooth-changing non-stationary bandits.

\begin{figure}[htbp!]
    \centering
\includegraphics[width=10cm]{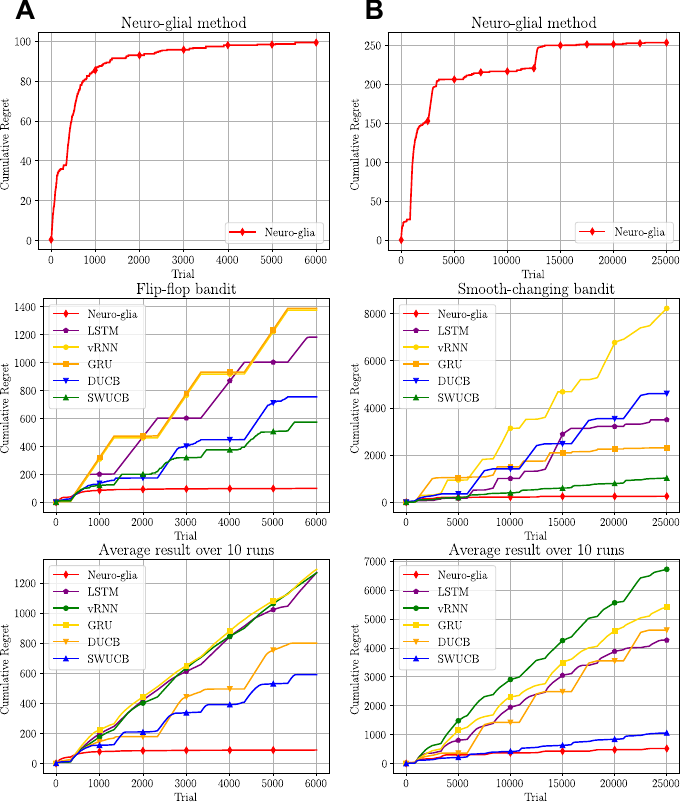}
    \caption{Learning performance in non-stationary Bernoulli bandits.  For flip-flop (panel A) and smooth-changing (panel B) cases, the cumulative regrets of different methods (Neuro-glia, LSTM, vRNN, GRU, DUCB, SWUCB) are shown for the single simulation and the average of $10$ runs. The hyperparameters of DUCB and SWUCB have been carefully tuned to optimize their performance.}
    \label{fig:extra_result_nonstationary_full}
\end{figure}

\subsubsection{Time-scale separation impact}\label{secsi_timescale}
The time-scale parameter has a mild impact on the learning performance for the stationary bandit. When $\tau=0.1$, the average cumulative regret is the smallest. It is also noted that the algorithm becomes more stable as the standard deviation of the final regrets becomes smaller as $\tau$ decreases.
The time-scale separation has important influence on the learning performance for the non-stationary bandit.  If there is no difference between the time-scales, i.e., $\tau=1$, the cumulative regrets will keep increasing and cannot reach a final stationary value for all runs, which means the agent is not able to adapt to the changing environments. When $\tau=0.1$, the agent can  achieve stationary cumulative regrets occasionally over multiple runs; if  $\tau \leq 0.01$, the agent can always achieve stationary asymptotic cumulative regret. 
\begin{figure}[htbp!]
    \centering
    \includegraphics[width=10cm]{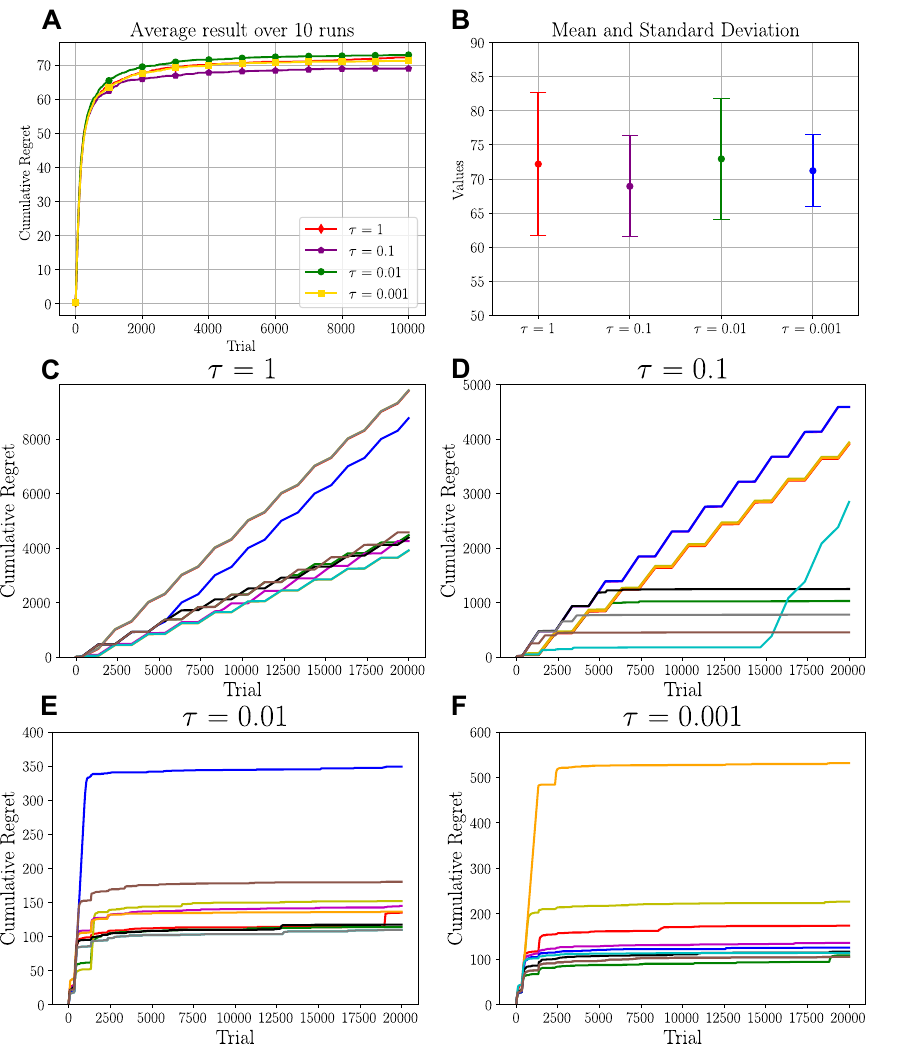}
    \caption{A is the average cumulative regrets over $10$ runs of the neuro-glial method with different $\tau$ in the stationary Bernoulli bandit, while B displays the mean and standard deviation of the asymptotic cumulative regrets. C-F show the cumulative regrets in each individual run of the neuro-glial method with different $\tau$ in the flip-flop non-stationary Bernoulli bandit.}
    \label{sifig:time-scale-impact}
\end{figure}

\subsubsection{Flexible generalization to bandit tasks with different number of actions}\label{secsi-8actions}
In previous demonstrations, we have simulated the bandit tasks with only $3$ actions. Here, we show that the neuro-glial networks and the designed learning algorithm can be easily generalized to other cases by providing an illustrative example involving an $8$-action Bernoulli bandit. 
In the stationary situation, the means for the actions are fixed as $\mu=(0.1,0.2,0.3,0.4,0.6,0.7,0.8,0.9)$. The non-stationary version is designed  with the means changing abruptly from $\mu$ to $1-\mu$ every $5000$ trials. 

To accommodate the $8$ actions of the bandit, the neuro-glial learning algorithm can be modified by simply expanding the dimension of the neuro-glial module's outputs to $8$, while leaving other settings unchanged.
It is validated from simulations that our method can solve these even challenging tasks very well in both stationary and non-stationary situations, and its learning performance is  superior in comparison with other methods (as shown in Figure \ref{sifig:extra_8actions}).

\begin{figure}[htbp!]
    \centering
\includegraphics[width=10cm]{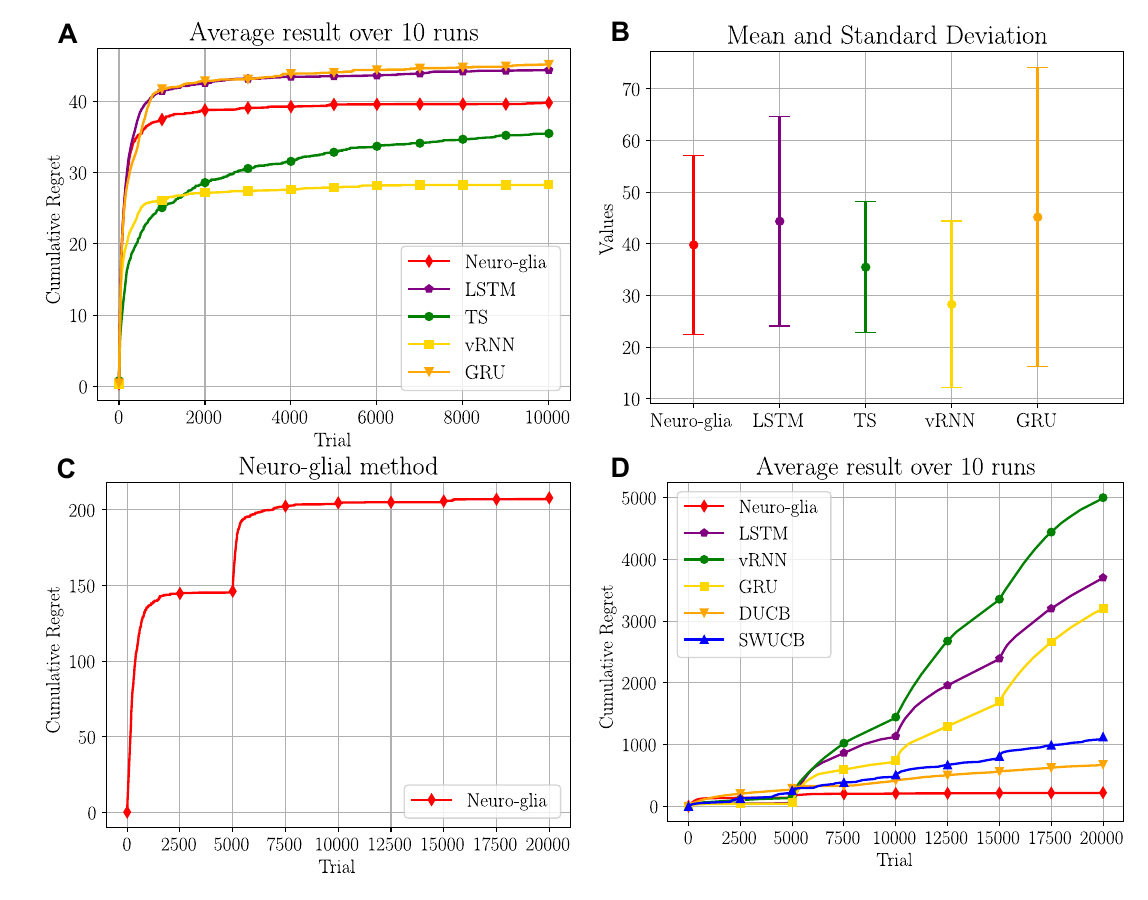}
    \caption{A. The plots show  the average results  and  the means and standard deviations of the asymptotic cumulative regret of different methods in the stationary bandit with $8$ actions.
B. The left plot is the result of a single simulation in this non-stationary bandit, while the right shows the average cumulative results of the neuro-glial method in comparison with other methods (again, DUCB and SWUCB methods have been tuned carefully).}
    \label{sifig:extra_8actions}
\end{figure}





\end{document}